\documentclass[a4paper,UKenglish,cleveref, autoref, thm-restate]{lipics-v2019}
\title{On the Parameterized Complexity of \textsc{Maximum Degree Contraction} Problem} %TODO Please add

\author{Saket Saurabh}{The Institute Of Mathematical Sciences, HBNI, Chennai, India \\ University of Bergen, Bergen, Norway}{saket@imsc.res.in}{}{
This project has received funding from the European Research Council
\begin{minipage}{0.6\textwidth}
(ERC) under the European Union's Horizon $2020$ research and innovation programme (grant agreement No $819416$), and Swarnajayanti Fellowship (No DST/SJF/MSA01/2017-18).
\end{minipage}
\begin{minipage}{0.3\textwidth}
        \includegraphics[scale=1]{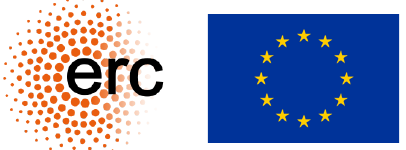}
\end{minipage}
}
\author{Prafullkumar Tale}{CISPA Helmholtz Center for Information Security, Saarbr$\ddot{\text{u}}$cken, Germany}{prafullkumar.tale@cispa.saarland}{}{This research is a part of a project that has received funding from the European Research Council (ERC) under the European Union's Horizon $2020$ research and innovation programme under grant agreement SYSTEMATICGRAPH (No. $725978$).}
\authorrunning{S. Saurabh and P. Tale}%TODO mandatory, please use full first names. LIPIcs license is "CC-BY";  http://creativecommons.org/licenses/by/3.0/


 %TODO mandatory; please add comma-separated list of keywords

\category{} %optional, e.g. invited paper

\relatedversion{} %optional, e.g. full version hosted on arXiv, HAL, or other respository/website
\supplement{}%optional, e.g. related research data, source code, ... hosted on a repository like zenodo, figshare, GitHub, ...

\Copyright{Saket Saurabh and Prafullkumar Tale}

\EventEditors{John Q. Open and Joan R. Access}
\EventNoEds{2}
\EventLongTitle{42nd Conference on Very Important Topics (CVIT 2016)}
\EventShortTitle{CVIT 2016}
\EventAcronym{CVIT}
\EventYear{2016}
\EventDate{December 24--27, 2016}
\EventLocation{Little Whinging, United Kingdom}
\EventLogo{}
\SeriesVolume{42}
\ArticleNo{23}
\usepackage{amsfonts}
\usepackage{complexity} % For symbols in complexity class.
\usepackage{todonotes} % For todo notes
\presetkeys{todonotes}{inline}{}

\newcommand{\calA}{\mathcal{A}}

\newcommand{\calB}{\mathcal{B}}

\newcommand{\calH}{{\mathcal H}}

\newcommand{\calO}{\ensuremath{{\mathcal O}}}
\newcommand{\OO}{\mathcal{O}}

\newcommand{\calU}{\mathcal{U}}

\newcommand{\calW}{\mathcal{W}}

\newcommand{\ETH}{\textsf{ETH}}

\newcommand{\CONP}{\textsf{coNP}}

\newcommand{\Hard}{\textsf{Hard}}

\newcommand{\yes}{\textsc{Yes}}
\newcommand{\no}{\textsc{No}}

\newcommand{\NO}{\textsc{No}}

\newtheorem{observation}{Observation}[section]
\newtheorem{reduction rule}{Reduction Rule}[section]
\newtheorem{branching rule}{Branching Rule}[section]
\newtheorem{marking-scheme}{Marking Scheme}[section]
\newcommand{\defparproblem}[4]{
  \vspace{1mm}
\noindent\fbox{
  \begin{minipage}{0.96\textwidth}
  \begin{tabular*}{\textwidth}{@{\extracolsep{\fill}}lr} #1  & {\bf{Parameter:}} #3
\\ \end{tabular*}
  {\bf{Input:}} #2  \\
  {\bf{Question:}} #4
  \end{minipage}
  }
  \vspace{1mm}
}

\begin{document}

\maketitle

%TODO mandatory: add short abstract of the document
\begin{abstract}
%!TEX root = main.tex
In the  \textsc{Maximum Degree Contraction} problem, input is a graph $G$ on $n$ vertices, and integers $k, d$, and the objective is to check whether $G$ can be transformed into a graph of  maximum degree at most $d$, using at most $k$ edge contractions. 
A simple brute-force algorithm that checks all possible sets of edges for a solution runs in time $n^{\mathcal{O}(k)}$.  As our first result, we prove that this algorithm is asymptotically optimal, upto constants in the exponents, under Exponential Time Hypothesis (\ETH). 

Belmonte,  Golovach,  van't Hof, and Paulusma  studied the problem in the realm of Parameterized Complexity and proved, among other things, that  it admits an \FPT\ algorithm running in time 
$(d + k)^{2k} \cdot n^{\mathcal{O}(1)} = 2^{\mathcal{O}(k \log (k+d) )} \cdot n^{\mathcal{O}(1)}$,  and remains \NP-hard for every constant $d \ge 2$ (Acta Informatica $(2014)$).
We present a different \FPT\ algorithm that runs in time $2^{\mathcal{O}(dk)} \cdot n^{\mathcal{O}(1)}$. In particular, our algorithm runs in time $2^{\mathcal{O}(k)} \cdot n^{\mathcal{O}(1)}$, for every fixed $d$. 
In the same article, the authors asked whether the problem admits a polynomial kernel, when parameterized by $k + d$.
We answer this question in the negative and prove that it does not admit a polynomial compression unless $\NP \subseteq \coNP/poly$.
\end{abstract}
\newpage
\setcounter{page}{1}
\section{Introduction}

For any graph class $\calH$, the $\calH$-\textsc{Modification} problem takes as input a graph $G$ and an integer $k$, and asks whether one can make at most $k$ modifications in $G$ such that the resulting graph is in $\calH$. 
These types of modification problems are one of the central problems in graph theory and have received a considerable attention in algorithm design.
With appropriate choice of $\calH$ and allowed modification operations, $\calH$-\textsc{Modification} can encapsulate well studied problems like \textsc{Vertex Cover}, \textsc{Chordal Completion}, \textsc{Cluster Editing}, \textsc{Hadwinger Number}, etc.
Some natural and well-studied graph modification operations are vertex deletion, edge deletion, edge addition, and edge contraction. 
The focus of the vast majority of papers on graph modification problems has been to the first three operations. 
Consider an example of $\calH_{\le d}$-\textsc{Modification} problem where $\calH_{\le d}$ is the collection of all graphs that has maximum degree at most $d$.
If allowed modification operation is vertex deletion then we know the problem as \textsc{Bounded Degree Deletion (BDD)} and if it is edge contraction then as \textsc{Maximum Degree Contraction (MDC)}.
The complexity of \textsc{BDD} and several of its variants has been extensively studied \cite{ balasundaram2010approximation, betzler2012bounded, bodlaender2001reduction, chen2010linear, dessmark1993maximum, fellows2011generalization, ganian2018structural, komusiewicz2009isolation, nishimura2005fast} whereas, to the best of our knowledge, only \cite{Belmonte:2014} addressed \textsc{MDC}.
In this article, we enhance our understanding of the second problem and answer an open question stated in \cite{Belmonte:2014}.

The {\em contraction} of edge $uv$ in simple graph $G$ deletes vertices $u$ and $v$ from $G$, and replaces them by a new vertex, which is made adjacent to vertices that were adjacent to either $u$ or $v$.
For a set of edges $F$ in $E(G)$, we denote the graph obtained from $G$ by contracting all edges in $F$ by $G/F$.
In the $\calH$-\textsc{Contraction} problem, an input is a graph $G$ and an integer $k$, and the aim is to decide whether there is a set $F$ of at most $k$ edges in $G$ such that $G/F$ is in $\calH$.
Early papers by Watanabe et al.~\cite{watanabe81,watanabe1983np} and Asano and Hirata~\cite{asano1983edge} showed that $\calH$-\textsc{Contraction} is \NP-\Hard\ for simple graph classes like trees, paths, stars, etc.
Brouwer proved that it is \NP-\Hard\ even to decide whether a graph can be contracted to a path of length four~\cite{brouwer1987contractibility}.
Note that this problem admits a simple polynomial time algorithm if we consider any other modification operation.
This has been a recurring theme in graph modification problems.
For the same target graph class, edge contraction problem tends to more difficult than their counterparts where modification operation is vertex/edge addition/deletion.
This difficulty is evident even in the realm of the Parameterized Complexity and Exact Exponential Algorithms.

In Parameterized Complexity, $\calH$-\textsc{Contraction} problems are studied with the number of edges allowed to contract, $k$, as parameter.
Heggernes et al.~\cite{heggernes2014contracting} proved that if $\calH$ is the set of acyclic graphs then $\calH$-\textsc{Contraction} is \FPT\ but does not admit a polynomial kernel unless $\NP \subseteq \CONP/poly$.
The vertex deletion version of the problem, known as \textsc{Feedback Vertex Set}, admits a polynomial kernel.
%Heggernes et al.~\cite{heggernes2014contracting} also proved that \textsc{Path-Contraction} admits a kernel with $\calO(k)$ vertices.
Series of papers studied the parameterized complexity for various graph classes like generalization and restrictions of trees~\cite{agarwal2019parameterized, agrawal2017paths},
cactus~\cite{krithika2018fpt},
bipartite graphs
~\cite{guillemot2013faster,heggernes2013obtaining}, planar graphs~\cite{golovach2013obtaining},
grids~\cite{saurabh2020parameterized},
cliques~\cite{cai2013contracting}, 
split graphs~\cite{agrawal2019split}, 
chordal graphs~\cite{lokshtanov2013hardness}, 
bi-cliques~\cite{martin2015computational},
degree constrained graph classes~\cite{Belmonte:2014,golovach2013increasing}, etc.
Krithika et al.~\cite{krithika2016lossy} and
Gunda et al.~\cite{gunda2020parameterized} studied $\calH$-\textsc{Contraction} problems from the lenses of \FPT\ approximation and lossy kernelization.
Agarwal et al.~\cite{agrawal2020path} broke the $2^n$-barrier for \textsc{Path Contraction} whereas Fomin et al.~\cite{fomin2020computation} showed that brute-force algorithms for \textsc{Hadwinger Number} problem and various other $\calH$-\textsc{Contraction} problem are optimal under \ETH.

Belmonte et al.~\cite{Belmonte:2014} studied the parameterized complexity of $\calH$-\textsc{Contraction} for three different classes $\calH$: the class of graphs with maximum degree at most $d$, the class of $d$-regular graphs, and the class of $d$-degenerate graphs.
They classified the parameterized complexity of all three problems with respect to the
parameters $k$, $d$, and $d + k$.
The first problem, also known as \textsc{MDC}, is defined as follows.

\defparproblem{\textsc{Maximum Degree Contraction} }{Graph $G$, integers $k, d$}{$k + d$}{Does there exist a subset $F$ of $E(G)$ of size at most $k$ such that every vertex in $G/F$ has degree at most $d$?}

The authors proved that \textsc{MDC} is \FPT\ when parameterized by $k + d$, $\W[2]$-\Hard\ when parameterized by $k$ (even when restricted to split graphs), and para-\NP-\Hard\ when parameterized by $d$.
Note that the problem is trivially solvable in polynomial time when $d \le 1$ and \NP-\Hard\ for every constant $d \ge 2$.

Consider brute-force algorithm for \textsc{MDC} that given an instance $(G, k, d)$, where graph $G$ has $n$ vertices, enumerates all subsets of edges of size at most $k$ in $G$ and for each subset contracts all edges in it to check whether the resulting graph has degree at most $d$.
This algorithm runs in time $n^{\calO(k)}$.
Our first results states that this algorithm is optimal, up to constants in the exponents, under \ETH.
\begin{theorem}\label{thm:brute-force-algo-lb} Unless \ETH\ fails, there is no algorithm that given any instance $(G, k, d)$ of \textsc{Maximum Degree Contraction} runs in time $n^{o(k)}$ and correctly determines whether it is a \yes\ instance.
\end{theorem}
Belmonte et al.~\cite{Belmonte:2014} presented an \FPT\ algorithm for \textsc{MDC} that runs in time $(d + k)^{2k} \cdot n^{\calO(1)}$.
As for any non-trivial instance $d + k$ is smaller than $n$, we can conclude that there is no algorithm that given any instance $(G, k, d)$ of \textsc{MDC} runs in time $(d + k)^{o(k)} \cdot n^{\calO(1)}$ and correctly determines whether it is a \yes\ instance, unless \ETH\ fails.

We remark that that the lower bound in Theorem~\ref{thm:brute-force-algo-lb} does not hold when $d$ is a fixed constant and not a part of input.
Hence, it is possible that \textsc{MDC} admits an algorithm that runs in time $k^{o(k)} \cdot n^{\calO(1)}$ for a constant value of $d$.
Belmonte et al.~\cite{Belmonte:2014} proved that \textsc{MDC} problem admits linear vertex kernels on connected graphs when $d = 2$.
This linear kernel leads to an \FPT\ algorithm\footnote{The algorithm colors vertices in the reduced instance with two colors and contracts each connected component in the colored subgraphs.} running in time $2^{\calO(k)} \cdot n^{\calO(1)}$ .
This hints that it is possible to design a better \FPT\ algorithm for small values of $d$.
Our second result shows that this is indeed the case.

\begin{theorem}\label{thm:fpt-new} There is an algorithm that given an instance $(G, k, d)$ of \textsc{Maximum Degree Contraction} runs in time $2^{\calO(d k)} \cdot n^{\calO(1)}$ and correctly determines whether it is a \yes\ instance.
\end{theorem}
We note that the reduction used in~\cite{Belmonte:2014} to prove that \textsc{MDC} is \NP-\Hard\ for any constant $d \ge 2$ implies that there is no $2^{o(dk)}$ algorithm for this problem.

Next, we look at the kernelization of \textsc{MDC}.
Belmonte et al.~\cite{Belmonte:2014} left it as an open question to determine whether \textsc{MDC} admits a polynomial kernel when parameterized by $k + d$.
Our last result answers this question in negative.
\begin{theorem}\label{thm:no-poly-kernel-MDC} Unless $\NP \subseteq \coNP/poly$, \textsc{Maximum Degree Contraction}, parameterized by $k + d$, does not admit a polynomial compression.
\end{theorem}
It is known that the \textsc{Bounded Degree Deletion} problem admits a kernel with $\calO(d^3 k)$ vertices \cite{fellows2011generalization}.
Hence, $\calH_{\le d}$-\textsc{Modification} is another example for which changing the modification operations from vertex deletion to edge contraction changes the compressibility drastically.

We organize the remaining paper as follows.
In Section~\ref{sec:prelims}, we present some preliminaries and observations regarding \textsc{MDC}.
In Section~\ref{sec:lower-bound}, we give a parameter preserving reduction from $(k \times k)$-\textsc{Permutation Independent Set} to \textsc{MDC} to rule out $n^{o(k)}$ algorithm for the later problem under \ETH.
We present an \FPT\ algorithm using universal sets and branching techniques in Section~\ref{sec:fpt-algo}.
In Section~\ref{sec:no-poly-kernel}, we present a parameter preserving reduction from \textsc{Red Blue Dominating Set} to rule out polynomial compression for \textsc{MDC} problem.
We conclude this article with an open question in Section~\ref{sec:conclusion}.
\section{Preliminaries}
\label{sec:prelims}
For a positive integer $q$, we denote set $\{1, 2, \dots, q\}$ by $[q]$.
\subsection{Graph Theory}
In this article, we consider simple graphs with a finite number of vertices.
For an undirected graph $G$, sets $V(G)$ and $E(G)$ denote its set of vertices and edges, respectively.
Unless otherwise specified, we use $n$ to denote the number of vertices in the input graph $G$.
We denote an edge with two endpoints $u, v$ as $(u, v)$.
Two vertices $u, v$ in $V(G)$ are \emph{ adjacent } to each other if there is an edge $(u, v)$ in $E(G)$. 
The open neighborhood of a vertex $v$, denoted by $N_G(v)$, is the set of vertices adjacent to $v$ and its degree $\deg_G(v)$ is $|N_G(v)|$.
The closed neighborhood of a vertex $v$, denoted by $N_G[v]$, is the set $N(v) \cup \{v\}$.
We omit the subscript in the notation for neighborhood and degree if the graph under consideration is clear.
For a subset $S$ of $V(G)$, we define $N[S] = \bigcup_{v \in S} N[v]$ and $N(S) = N[S] \setminus S$.
%For a graph $G$, we denote the maximum degree of any vertex in graph $G$ by $\Delta(G)$. 
For a subset $F$ of edges, a subset of vertices $V(F)$ denotes the collection of endpoints of edges in $F$.
We say a set of edges $F$ \emph{spans} a set of vertices $S$ if $S \subseteq V(F)$.
For a subset $S$ of $V(G)$, we denote the graph obtained by deleting $S$ from $G$ by $G - S$ and the subgraph of $G$ induced on the set $S$ by $G[S]$. 
For two subsets $S_1, S_2$ of $V(G)$, edge set $E(S_1, S_2)$ denotes the edges with one endpoint in $S_1$ and another one in $S_2$. 
We say $S_1, S_2$ are adjacent if $E(S_1, S_2)$ is non empty. 
For an integer $q$, a $q$-coloring of graph $G$ is a function $\phi : V(G) \rightarrow [q]$.
A \emph{proper coloring} of $G$ is a $q$-coloring $\phi$ of $V(G)$ for some integer $q$ such that for any edge $(u, v)$, $\phi(u) \neq \phi(v)$. 
There is a proper coloring of the graph with $\Delta(G) + 1$ many colors which can found in polynomial time.
A set of vertices $S$ is said to be \emph{independent set} if no two vertices in $S$ are adjacent to each other.
A set of edges $F$ is called \emph{matching} if no two edges in $F$ share an endpoint.
A graph is called {\em connected} if there is a path between every pair of distinct vertices.
A subset $S$ of $V(G)$ is said to be a \emph{connected set} if $G[S]$ is connected.
A \emph{spanning tree} of a connected graph is its connected acyclic subgraph, which includes all the vertices of the graph.

\subsection{Graph Contraction}
The {\em contraction} of an edge $uv$ in $G$ deletes vertices $u$ and $v$ from $G$, and adds a new vertex which is adjacent to vertices that were adjacent to either $u$ or $v$.
This process does not introduce self-loops or parallel edges.
The resulting graph is denoted by $G/e$.
For a graph $G$ and edge $e = uv$, we formally define $G/e$ in the following way: $V(G/e) = (V(G) \cup \{w\}) \backslash \{u, v\}$ and $E(G/e) = \{xy \mid x,y \in V(G) \setminus \{u, v\}, xy \in E(G)\} \cup \{wx |\ x \in N_G(u) \cup N_G(v)\}$.
Here, $w$ is a new vertex.
An edge contraction reduces the number of vertices in a graph by exactly one.
Several edges might disappear because of one edge contraction. 
For a subset of edges $F$ in $G$, graph $G/ F$ denotes the graph obtained from $G$ by contracting each connected component in the sub-graph $G’ = (V(F), F)$ to a vertex.

We now formally define a contraction of graph $G$ to another graph $H$.
\begin{definition}[Graph Contraction] \label{def:graph-contractioon} A graph $G$ is said to be \emph{contractible} to graph $H$ if there is a function $\psi: V(G) \rightarrow V(H)$ such that following properties hold.
\begin{enumerate}
\item For any vertex $h$ in $V(H)$, set $W(h) := \{v \in V(G) \mid \psi(v)= h\}$ is not empty and graph $G[W(h)]$ is connected.
\item For any two vertices $h, h’$ in $V(H)$, edge $hh’$ is present in $H$ if and only if $E(W(h), W(h’))$ is not empty.
\end{enumerate}
\end{definition}
We say graph $G$ is contractible to $H$ via mapping $\psi$.
For a vertex $h$ in $H$, set $W(h)$ is called a \emph{witness set} associated with or corresponding to $h$. 
We define the $H$-\emph{witness structure} of $G$, denoted by $\mathcal{W}$, as a collection of all witness sets.
Formally, $\mathcal{W}=\{W(h) \mid h \in V(H)\}$.
A witness structure $\mathcal{W}$ is a partition of vertices in $G$. 
If a \emph{witness set} contains more than one vertex, then we call it \emph{big} witness set, otherwise it is \emph{small} witness set.

If graph $G$ has a $H$-witness structure, then graph $H$ can be obtained from $G$ by a series of edge contractions.
For a fixed $H$-witness structure, let $F$ be the union of spanning trees of all witness sets.
By convention, the spanning tree of a singleton set is the empty set.
To obtain graph $H$ from $G$, it is sufficient to contract edges in $F$.
Hence, $H = G/F$.
For a $G/F$-witness structure $\calW$ of $G$, there is a unique function $\psi : V(G) \rightarrow V(G/F)$ corresponding to it.
We say graph $G$ is \emph{$k$-contractible} to $H$ if the cardinality of $F$ is at most $k$.
In other words, $H$ can be obtained from $G$ by at most $k$ edge contractions.
The following observations are immediate consequences of definitions.
\begin{observation}
\label{obs:witness-structure-property} If graph $G$ is $k$-contractible to graph $H$ via mapping $\psi$ then following statements are true. 
\begin{enumerate} 
%\item For any witness set $W$ in a $H$-witness structure of $G$, cardinality of $W$ is at most $k+1$.
\item \label{item:nr-big-witness-set} Any $H$-witness structure of $G$ has at most $k$ big witness sets.
\item \label{item:nr-vertices-big-bag} For a fixed $H$-witness structure, the number of vertices in $G$ which are contained in big witness sets is at most $2k$.
\item \label{item:degree-bound} For a vertex $v$ in $V(G)$, if $|W(\psi(v))| = 1$ then $ \deg_{H}(\psi(v)) \le \deg_{G}(v)$.
\item \label{item:reduced-size} For $U \subseteq V(G)$, define $\psi(U) := \{\psi(u)\ |\ u \in U \}$. Then, $|U| \leq |\psi(U)|+ k$.
\end{enumerate}
\end{observation}
\begin{proof}
Let $\calW$ be the $H$-witness structure of $G$ and $F$ be the union of the spanning trees of all witness sets.
As $G$ is $k$-contractible to $H$, we have $|F| \le k$.

\noindent (\ref{item:nr-big-witness-set}) As any big witness set contains at least one edge in $F$, the number of big witness set is at most $k$.

\noindent (\ref{item:nr-vertices-big-bag}) As $F$ spans all vertices in big witness set, the number of vertices in big witness set is at most $2k$.

\noindent (\ref{item:degree-bound})
Let $h_i$ be a vertex in $N_H(\psi(v))$.
As $(\psi(v), h_i) \in E(H)$, set $E(W(\psi(v)), W(h_i))$ is a non-empty subset of $E(G)$.
As $|W(\psi(v))| = 1$, this implies $E(\{v\}, W(h_i))$ is a non empty.
As $h_i$ is an arbitrary neighbor of $\psi(v)$, we can conclude that in graph $G$, $v$ is adjacent with at least as many vertices as $\deg_H(\psi(v))$.
Hence, $\deg_H(\psi(v)) \le \deg_G(v)$.

\noindent (\ref{item:reduced-size}) Assume that $|U| > |\psi(U)| + k$.
Fix an arbitrary order on vertices in $U$.
We define a function $\phi: U \rightarrow \psi(U) \cup \{\bot\}$ as follows: $\phi(u_i) = \phi(u_i)$ if $\phi(u_i) \neq \phi(u_j)$ for $j < i$ otherwise $\phi(u_i) = \bot$.
For a vertex $\psi(u_i)$ in $\psi(U)$, the function $\phi$ selects one vertex amongst the set $\{u |\ \psi(u) = \psi(u_i)\}$. 
Define $U_0 = \{ u \in U |\ \phi(u) = \bot \}$.
By our assumption, $|U_0| > k$. 

Consider an arbitrary vertex $u_i$ in $U_0$. 
By the construction, there is an index $j \in [|U|]$ such that $u_j \in U$, $\psi(u_j) = \psi(u_i)$ and $j < i$.
As $\psi(u_i) = \psi(u_j)$, both $u_i, u_j$ are in some big witness set in $\calW$.
As $F$ is the union of edges of spanning trees of witness sets in $\calW$, there is a unique path from $u_i$ to $u_j$ that comprises only edges in $F$.
Consider the edge in this path incident to $u_i$.
We assign vertex $u_i$ to this edge in $F$.
As $u_i$ is an arbitrary vertex in $U_0$, we can assign an edge in $F$ to every vertex in $U_0$.
Note that we are considering the first edge in the unique path from some vertex in $U_0$ to some vertex outside $U_0$.
Hence, no two vertices in $U_0$ can be assigned to same edge in $F$.
This contradicts the fact that $|F| \le k$.
Hence, our assumption is wrong and $|U| \le |\psi(U)| + k$.
\end{proof}

\subsection{\textsc{Maximum Degree Contraction}}
In this subsection, we prove some observations and a lemma related to \textsc{MDC}.
We say a set of edges $F$ is a \emph{solution} to instance $(G, k, d)$ if the number of edges in $F$ is at most $k$ and the maximum degree of graph $G/F$ is at most $d$.
The number of edges that we are allowed to contract, $k$, is also called \emph{solution size}. 
We start with the following simple observation that states that contracting an edge in a solution does not produce a \no\ instance.
\begin{observation}\label{obs:reduced-inst-yes}
If $(G, k, d)$ is a \yes\ instance of \textsc{MDC} and $F \subseteq E(G)$ is a solution to $(G, k, d)$, then for any edge $(u, v)$ in $F$, instance $(G/\{(u,v)\}, k - 1, d)$ is a \yes\ instance of \textsc{MDC}.
\end{observation}
We bound the maximum degree of graph by $k + d$ in the non-trivial instances of the problem.
\begin{observation} \label{obs:trivial-no} If there is a vertex of degree $d + k +1$ or more in $G$ then $(G, k, d)$ is a \no\ instance.
\end{observation}
\begin{proof}
Suppose there is a vertex, say $v$, of degree greater than $d + k + 1$ in graph $G$.
Assume, for the sake of a contradiction, that $(G, k, d)$ is a \yes-instance. 
Let $(G, k, d)$ is $k$-contractible to a graph $H$, via mapping $\psi$, such that the maximum degree of vertices in $H$ is at most $d$.
By Observation~\ref{obs:witness-structure-property} (\ref{item:reduced-size}), $|N_G[v]| \le |\psi(N_G[v])| + k$ where $\psi(N_G[v]) = \bigcup_{u \in N_G[v]} \psi(u)$.
As $d + k + 1 < |N_G[v]|$, we have $d + 1 < |\psi(N_G[v])|$.
As $\psi(N_G[v]) \subseteq N_H(\psi(v))$, vertex $\psi(v)$ is adjacent with $d + 1$ or more vertices in $H$.
This contradicts the fact that the maximum degree of vertices in $H$ is at most $d$.
Hence, our assumption was wrong and $(G, k, d)$ is a \no\ instance.
\end{proof}
If every vertex in $G$ has degree at most $d$, then $(G, k, d)$ is a trivial \yes\ instance.
Hence, there is at least one vertex in $G$ that has degree at least $d + 1$.
We prove that the number of such vertices is bounded.
\begin{observation}\label{obs:bound-large-deg-vertices} Let $(G, k, d)$ be a \yes\ instance of \textsc{MDC}. Then, $G$ contains at most $k(d + 2)$ vertices that has degree at least $d + 1$.
\end{observation}
\begin{proof}
Let $L$ be the collection of vertices in $G$ which has degree at least $d + 1$.
As $(G, k, d)$ is a \yes\ instance, there is a solution, say $F$, to it.
Let $\calW$ be a $(G/F)$-witness structure of $G$.
By Observation~\ref{obs:high-deg-vertex}, every vertex in $L$ is either contained in a big-witness set or at least two of its neighbors are in a big witness set.
By Observation~\ref{obs:witness-structure-property} (\ref{item:nr-vertices-big-bag}), the number of vertices in big witness sets is at most $2k$.
As every vertex in $G/F$ has degree at most $d$, there are at most $dk$ vertices in $G$ that are adjacent to some vertex in big witness sets.
This implies that there are at most $k(2 + d)$ vertices in $L$.
\end{proof}
The following observation specifies how a solution behaves locally. 
\begin{observation} \label{obs:high-deg-vertex}
Consider a \yes\ instance $(G, k, d)$ of \textsc{MDC} and let $v$ be a vertex of degree at least $d + 1$ in $G$.
Then, for any solution $F$ to $(G, k, d)$, there are at least two vertices in $N[v]$ that are in the same witness set in the $G/F$-witness structure of $G$.
\end{observation}
\begin{proof}
Let $G$ is contractible to a graph $G/F$, via mapping $\psi$.
Assume, for the sake of contradiction, that no two vertices in $N[v]$ are in the same witness set.
This implies $|N[v]| = |\psi(N[v])|$, where $\psi(N_G[v]) = \bigcup_{u \in N_G[v]} \psi(u)$.
As $\psi(N_G[v]) \subseteq N_{G/F}(\psi(v))$ and $|N[v]| > d + 1$, vertex $\psi(v)$ is adjacent with $d + 1$ or more vertices in $G/F$.
This contradicts the fact that the maximum degree of vertices in $G/F$ is at most $d$.
Hence, our assumption was wrong and there are at least two vertices in $N[v]$ that are in some big-witness set in $G/F$-witness structure of $G$.
\end{proof}

We say that solution $F$ \emph{merges} at least two vertices in $N[v]$.
Note that for an edge $(u, v)$ in $G$, it is possible that $(u, v) \not\in F$ but $F$ merges $u, v$.

The following lemma allows us to conclude that an instance is a \no\ instance if we find a sizeable collection of \emph{large stars} that do not intersect with each other.
We present it in the form suitable for the application in the later part of the article. 
\begin{lemma} \label{lemma:disjoint-nbrs} For an instance $(G, k - 1, d)$, suppose there is subset $L^{\circ}$ of $V(G)$ that satisfies the following conditions:
$(i)$ For every vertex $v$ in $L^{\circ}$, $N(v)$ is an independent set of size at least $d + 1$.
$(ii)$ For any two different vertices $u, v$ in $L^{\circ}$, $N(v) \cap N(u) = \emptyset$. 
$(iii)$ $|L^{\circ}| \ge k$.
Then, $(G, k - 1, d)$ is a \no\ instance.
\end{lemma}
\begin{proof}
Assume, for the sake of contradiction, that $(G, k - 1, d)$ is a \yes\ instance.
Let $F$ be a solution to $(G, k, d)$ and $G$ is $(k-1)$-contractible to be $G/F$ via $\psi$.
By Observation~\ref{obs:high-deg-vertex}, for every vertex $v$ in $L^{\circ}$, there are at least two vertices in $N[v]$ which are in same witness set in the $G/F$-witness structure of $G$.
As $N(v)$ is an independent set, there is no edge whose both endpoints are in $N(v)$.
Hence, for every vertex in $L^{\circ}$, one of the following two statements must be true: 
$(a)$ $F$ contains an edge incident to $v$. 
$(b)$ $F$ contains at least two edges incident to $N(v)$ but are not incident to $\{v\}$.
Let $L_1$ be the collection of vertices in $L^{\circ}$ for which the first statement is true. 
Let $F_1$ be the subset of $F$ that are incident to some vertex in $L_1$. 
Recall that for any two vertices in $u, v \in L_1$, as $N[v] \cap N[u] = \emptyset$.
Hence, no edge in $F_1$ is incident to more than one vertex in $L^{\circ}$.
Hence, $|F_1| = |L_1|$.
For every $v$ in $L \setminus L_1$, there are at least two edges incident to $N(v)$.
Note that these edges are in $F \setminus F_1$ as they are not incident to any vertex in $L^{\circ}$.
As every edge in $F \setminus F_1$ can be incident to the open neighborhood of at most two vertices in $L \setminus L_1$, we can conclude that $2 |F \setminus F_1| \ge 2 |L \setminus L_1|$.
This implies that the number of edges in $F$ is at least $|L_1| + |L \setminus L_1| = |L|$.
This contradicts the fact that $|F| \le k - 1$ and $|L| \ge k + 1$.
Hence our assumption is wrong and $(G, k, d)$ is a \no\ instance.
\end{proof}

\subsection{Parameterized Complexity} 
An instance of a parameterized problem $\Pi$ comprises of an input $I$, which is an input of the classical instance of the problem and an integer $k$, which is called as the parameter.
A problem $\Pi$ is said to be \emph{fixed-parameter tractable} or in \FPT\ if given an instance $(I,k)$ of $\Pi$, we can decide whether or not $(I,k)$ is a \yes\ instance of $\Pi$ in  time $f(k)\cdot |I|^{\OO(1)}$.
Here, $f(\cdot)$ is some computable function whose value depends only on $k$. 

A \emph{compression} of a parameterized problem $\Pi_1$ into a (non-parameterized) problem $\Pi_2$ is a polynomial algorithm that maps each instance $(I_1, k_1)$ of $\Pi_1$ to an instance $I$ of $\Pi_2$ such that $(1)$ $(I, k)$ is a \yes\ instance of $\Pi_1$ if and only if $I_2$ is a \yes\ instance of $\Pi_2$, and $(2)$ size of  $I_2$ is bounded by $g(k)$ for a computable function $g(\cdot)$.
The output $I_2$ is also called a compression. 
The function $g$ is said to be the size of the compression.
A compression is polynomial if $g$ is polynomial.
A kernelization algorithm for a parameterized problem $\Pi$ is a polynomial algorithm that maps each instance $(I, k)$ of $\Pi$ to an instance $(I', k')$ of $\Pi$ such that $(1)$ $(I, k)$ is a \yes\ instance of $\Pi$ if and only if $(I', k')$ is a \yes\ instance of $\Pi$, and $(2)$ $|I'| + k'$ is bounded by $g(k)$ for a computable function $g(\cdot)$.
Respectively, $(I', k')$ is a kernel and $g$ is its size. 
A kernel is polynomial if $g$ is polynomial.

It is typical to describe a compression or kernelization algorithm as a series of reduction rules.
A reduction rule is a polynomial algorithm that takes as an input an instance of a problem and output another, usually reduced, instance.
A reduction rule said to be \emph{applicable} on an instance if the output instances is different than the instance.
A reduction rule is \emph{safe} if the input instance is a \yes\ instance if and only if the output instance is a \yes\ instance.

For details on parameterized complexity and related terminologies, we refer the reader to the books of Downey and Fellows~\cite{DF-new}, Flum  and Grohe~\cite{flumgrohe}, Niedermeier~\cite{niedermeier2006}, and the more recent books by Cygan et al.~\cite{saurabh-book} and Fomin et al.~\cite{fomin2019kernelization}.
\section{A Lower Bound for the Algorithm}
\label{sec:lower-bound}
In this section, we prove Theorem~\ref{thm:brute-force-algo-lb}. 
We present a reduction from $(k \times k)$-\textsc{Permutation Independent Set (PIS)} problem to \textsc{Maximum Degree Contraction} problem.
In the $(k \times k)$-\textsc{PIS} problem we are given a graph $H$ on a vertex set $[k] \times [k]$.
In other words, the vertex set is formed by a $k \times k$ table. 
We denote vertices in the table by $v[i, j]$ for $1 \le i, j \le k$. 
The question is whether there exists an independent set $X$ in $H$ that contains exactly one vertex from each row and each column of the table.
In other words, for every $i, j \in [k]$ there is exactly one element of $X$ that has $i$ on
the first coordinate and $j$ on the second coordinate. 
Note that without loss of generality we may assume that each row and each column of the table forms an independent set.~\footnote{Since we are looking for an independent set, it is intuitive to add all missing edges in a row or a column of the table. But to simply our reduction, we remove edges that have both their endpoints in the same row or column. It is easy to verify that this operation is safe.}
The following result is known for this problem.

\begin{proposition}[\cite{lokshtanov2018slightly}] \label{prop:ind-set-lower-bound} Unless \ETH\ fails, $(k \times k)$-\textsc{Permutation Independent Set} can not be solved in time $k^{o(k)}$.
\end{proposition}

\paragraph*{Reduction}
The reduction accepts an instance, say $(H, k)$, of $(k \times k)$-\textsc{Permutation Independent Set} as an input.
Here, $H$ is a graph with vertex set formed by a $k \times k$ table.
The reduction modifies a copy of the graph $H$ in the following way. 
\begin{itemize}
\item[-] It adds a vertex corresponding to each row in the table and makes it adjacent with all vertices in that row.
Let $R = \{r_1, r_2, \dots, r_k\}$ be the set of vertices corresponding to rows. 
\item[-] It adds a vertex corresponding to each column in the table and makes it adjacent with all vertices in that column.
Let $C = \{c_1, c_2, \dots, c_k\}$ be the set of vertices corresponding to columns.
\item[-] It adds set $S = \{s_1, s_2, \dots, s_k\}$ of $k$ vertices.
For every $i$ in $[k]$, it makes $s_i$ adjacent with every vertex in $V(H) \cup C$ and with $r_i$. 
\item[-] For every vertex $r_i$ in $R$, it adds $k^2$ pendant vertices and makes them adjacent with $r_i$. 
\item[-] For every vertex $c_j$ in $C$, it adds $(k^2 - k + 1)$ pendant vertices and makes them adjacent with $c_j$.
\end{itemize}
See Figure~\ref{fig:reduction-IS-to-Max-Deg-Cont} for an illustration.
Let $G$ be the graph obtained from a copy of graph $H$ with the above modifications.
The algorithm returns $(G, k, k^2 + k)$ as instance of \textsc{MDC}.

\begin{figure}[t]
\begin{center}
\includegraphics[scale=0.75]{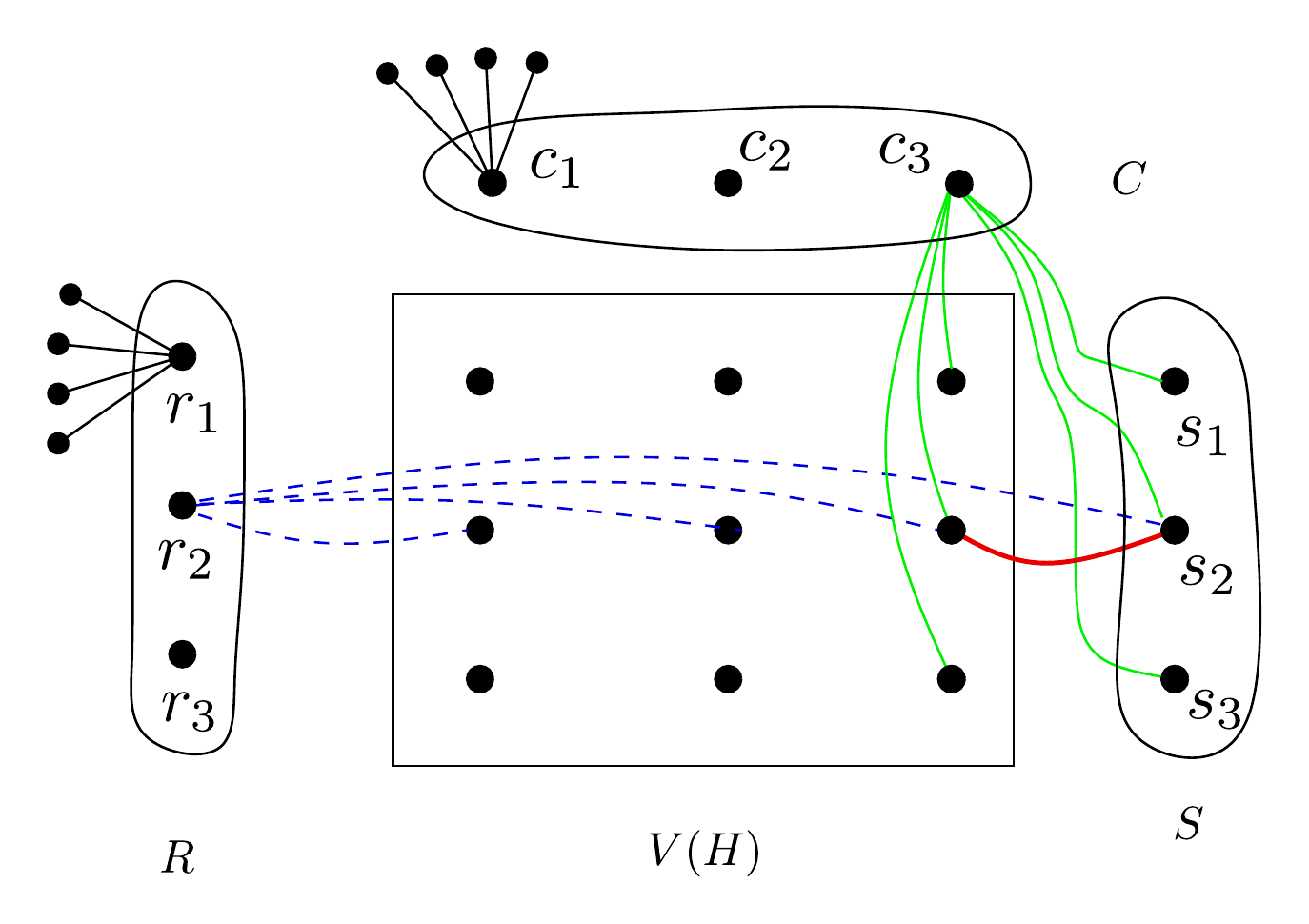}
\end{center}
\caption{Dotted (blue) lines and thin (green) lines show the adjacency of vertices in $R$ and $C$, respectively. 
Contracting the thick (red) edge $(v[2, 3], s_2)$ represents selecting vertex $v[2, 3]$ into the independent set.
For the sake of clarity, we do not depict all edges present in the graph. \label{fig:reduction-IS-to-Max-Deg-Cont}}
\end{figure}

We present intuition of the proof of correctness.
We describe how a solution, if it exists, to $(G, k, d)$ leads to a solution to $(H, k)$.
We hope that this will also provide some intuition as to how a solution to $(H, k)$ leads to a solution to $(G, k, d)$.
Note that $S, C, R$ are independent sets in $G$.
Every vertex in $R \cup C \cup S$ has degree $d + 1$ and every vertex in $V(G) \setminus (R \cup C \cup S)$ has degree strictly less than $d$.
We first argue that any solution for $(G, k, d)$ can only contain edges in $E(G)$ that have one endpoint in $V(H)$ and another endpoint in $S$.
Then, we prove that for every $i\in [k]$ a solution must pick an edge incident to some vertex in $i^{th}$ row and on $s_i$ to reduce the degree of vertex $r_i$.
We prove a similar statement for every column.
Hence, for every $i\in [k]$, a solution contains an edge of the form $(v[i, j], s_i)$ for some $j \in [k]$.
As there are at most $k$ edges in a solution, every edge is of this form. 
For $i_1, i_2, j_1, j_2 \in [k]$, let $(v[i_1, j_1], s_{i_1})$ and $(v[i_2, j_2], s_{i_2})$ be two edges in a solution.
We argue that if $(v[i_1, j_1], v[i_2, j_2])$ is an edge in $G$ (and hence in $H$) then degrees of vertices obtained by contracting 
$(v[i_1, j_1], s_{i_1})$ and $(v[i_2, j_2], s_{i_2})$ are more than $d$.
As this is true for any two arbitrary edges in solution, their endpoints in  $V(H)$ form an independent set in $H$.
We formalize this intuition in the following two lemmas.

\begin{lemma} \label{lemma:correct-reduction-PermIS-MaxDegCont-forward} Suppose the reduction returns $(G, k, d)$ when the input is $(H, k)$.
If $(H, k)$ is a \yes\ instance of $(k \times k)$-\textsc{PIS} then $(G, k, d)$ is a \yes\ instance of \textsc{MDC}.
\end{lemma}
\begin{proof}
Suppose $(H, k)$ is a \yes\ instance and let $X$ be an independent set in $H$ that contains exactly one vertex from each row and each column of the table.
Define a function $\rho : [k] \rightarrow [k]$ such that $j = \rho(i)$ if $v[i, j]$ is a vertex in $X$. 
By the properties of $X$, we can conclude that $\rho$ is one-to-one and onto function.
We construct solution $F$ to $(G, k, d)$ using independent set $X$.
For every vertex $v[i, \rho(i)]$ in $X$, add edge $(s_i, v[i, \rho(i)])$ in $F$.
By construction, the cardinality of $F$ is $k$.
We argue that the maximum degree of any vertex in $G/F$ is $d$.
As mentioned before, in graph $G$, set $R \cup C \cup S$ is the collection of all vertices of degree strictly greater than $d$. 
More precisely, every vertex in $R \cup C \cup S$ has degree $d + 1$.
We demonstrate that contracting edges in $F$ reduces the degree of each vertex in $R \cup C \ \cup S$ by one.

Note that edges in $F$ form a matching in $G$.
For every $i$ in $[k]$, let $s^\circ_i$ be the new vertex added while contracting edge $(s_i, v[i, \rho(i)])$.
Let $G' = G/F$ and $S^{\circ} = \{s^\circ_1, s^\circ_2, \dots, s^\circ_k\}$. 
Note that $V(G')$ can be partitioned into $R, C, S^\circ$, $V(H) \setminus X$, and pendent vertices which are adjacent with $R \cup C$.
Every vertex in $V(H) \setminus X$ is adjacent with at most $k^2 - 2k + k + 2 \le d - 2k + 2 \le d$ vertices in $G'$.
Every vertex in $R$ is adjacent with $k - 1$ vertices in $V(H) \setminus X$, one vertex in $S^{\circ}$, and $k^2$ pendant vertices in $G'$.
Hence, degree of every vertex in $R$ in $G'$ is $k - 1 + 1 + k^2 = d$.
For a vertex, say $c_j$, in $C$, there exists a vertex $v[i, \rho(i)]$ in $X$ such that $j = \rho(i)$.
Hence, in graph $G'$, vertex $c_j$ is adjacent with $k-1$ vertices in $V(H)$, $k$ vertices in $S^\circ$, and $k^2 - k + 1$ pendent vertices.
Hence, degree of $c_j$ in $G'$ is $k - 1 + k + k^2 - k + 1 = d$. 
Since $c_j$ is an arbitrary vertex in $C$, this is true for every vertex in $C$. 

We now argue that $S^\circ$ is an independent set in $G'$.
Consider two vertices, say $s^{\circ}_i, s^{\circ}_{j}$ in $S^\circ$.
By construction, vertices $s_i$ and $s_j$ are not adjacent with each other in $G$.
As $X$ is an independent set in $G$, vertices $v[i, \rho(i)], v[j, \rho(j)]$ in $X$ are not  adjacent with each other.
This implies there is no edge with one endpoint in set $\{s_i, v[i, \rho(i)]\}$ and another endpoint in $\{s_j, v[j, \rho(j)]\}$.
This implies that vertices $s^\circ_i$ and $s^\circ_j$ are not adjacent with each other in $G'$.
Since this is true for any two vertices in $S^\circ$, it is an independent set in $G'$.
By construction, for any $i$ in $[k]$, vertex $v[i, \rho(i)]$ is adjacent with only one vertex, viz $r_i$, in $R$.
Hence, any vertex $s^\circ_i$ in $S$ is adjacent with one vertex vertex in $R$, $k$ vertices in $C$ and $k^2 - 1$ vertices in $V(H) \setminus X$ in graph $G'$. 
This implies that every vertex in $S^\circ$ has degree $1 + k + k^2 - 1 = d$.
Hence, the maximum degree of any vertex in $G'$ is at most $d$.
This implies that $(G, k, d)$ is a \yes\ instance which concludes the proof.
\end{proof}

In the remaining section, we prove the following lemma.

\begin{lemma} \label{lemma:correct-reduction-PermIS-MaxDegCont-backword} Suppose the reduction returns $(G, k, d)$ when the input is $(H, k)$.
If if $(G, k, d)$ is a \yes\ instance of \textsc{MDC} then $(H, k)$ is a \yes\ instance of $(k \times k)$-\textsc{PIS}.
\end{lemma}

To prove the lemma, we first investigate how a solution to $(G, k, d)$ can intersect with edges in $G$.
Recall that for vertex subset $X, Y$, we denote the set of all edges with one endpoint in $X$ and another endpoint in $Y$ by $E(X, Y)$.
Let $P_C$ and $P_R$ be the collection of pendant vertices that are adjacent with $C$ and $R$, respectively.
By construction, edges of $G$ can be partitioned into following five sets: 
$E(C \cup R, P_C \cup P_R)$, $E(C \cup R, S)$, $E(C \cup R, V(H))$, $E(V(H), V(H))$, and $E(S, V(H))$.
We first prove that any solution to $(G, k, d)$ does not intersect with the first four sets.

Suppose $(G, k, d)$ is a \yes\ instance and $F \subseteq E(G)$ be a solution to $(G, k, d)$.
\begin{claim}
\label{claim:IS-backword-start}
$F \cap E(C \cup R, P_C \cup P_R) = \emptyset$.
\end{claim}
\begin{proof} 
Assume that there exist an edge, say $(c_i, v)$, in $F \cap E(C, P_C)$ where vertices $c_i, v$ are in $C$ and $P_C$, respectively.
Note that, instance $(G/\{(c_i, v)\}, k - 1, d)$ and set $R$ satisfy the premise of Lemma~\ref{lemma:disjoint-nbrs}.
Hence, we can conclude that $(G/\{(c_i, v)\}, k - 1, d)$ is a \no\ instance.
This contradicts Observation~\ref{obs:reduced-inst-yes}.
Hence our assumption is wrong and $F \cap E(C, P_C)$ is an empty set.

Assume that there exist edge $(r_i, v)$ in $F \cap E(R, P_R)$ where vertices $r_i, v$ are in $R$ and $P_R$, respectively.
Let $R^{\prime}$ be the set obtained from $R$ by removing $r_i$ and adding the vertex which was introduced while contracting edge $(r_i, v)$.
Note that, instance $(G/\{(r_i, v)\}, k - 1, d)$ and set $R^{\prime}$ satisfy the premise of Lemma~\ref{lemma:disjoint-nbrs}.
Hence, we can conclude that $(G/\{(r_i, v)\}, k - 1, d)$ is a \no\ instance.
This contradicts Observation~\ref{obs:reduced-inst-yes}.
Hence our assumption is wrong and $F \cap E(R, P_R)$ is an empty set.

By the construction of $G$, sets $E(C, P_R)$ and $E(R, P_R)$ are empty.
This implies that there is no edge in $F \cap E(C \cup R, P_C \cup P_R)$. 
\end{proof}

\begin{claim} \label{claim:IS-backward-second}
$F \cap E(C \cup R, S) = \emptyset$.
\end{claim}
\begin{proof} 
Assume that there exist an edge, say $(c_i, s_j)$, in $F \cap E(C, S)$ where vertices $c_i, s_j$ are in $C$ and $S$, respectively.
Let $w$ be the new vertex introduced while contracting edge $(c_i, s_j)$.
In graph $(G/\{(c_i, s_j)\}$, vertex $w$ is adjacent with every vertex in $(S \setminus \{s_j\}) \cup V(H) \cup (C \setminus \{c_i\})$ and with all pendent vertices which were adjacent with $c_i$ in $G$.
Hence, the degree of $w$ in $(G/\{(c_i, s_j)\})$ is at least $2k^2 + k  - 1 \ge (k^2 + k) + (k - 1) + 1 = d + (k - 1) + 1$.
By Observation~\ref{obs:high-deg-vertex}, $(G/\{(c_i, s_j)\}, k - 1, d)$ is a \no\ instance. 
This contradicts Observation~\ref{obs:reduced-inst-yes}.
Hence our assumption is wrong and $F \cap E(C, S)$ is an empty set.
We can conclude that $F \cap E(C, S)$ is an empty set by a similar argument. 
This concludes the proof of the claim. 
\end{proof}

\begin{claim}
$F \cap E(C, V(H)) = \emptyset$.
\end{claim}
\begin{proof} 
Assume that there exist an edge, say $(c_j, x_{ij})$, in $F \cap E(C, V(H))$ where vertices $c_j, x_{ij}$ are in $C$ and $V(H)$, respectively.
Note that, instance $(G/\{(c_j, x_{ij})\}, k - 1, d)$ and set $R$ satisfy the premise of Lemma~\ref{lemma:disjoint-nbrs}.
Hence, we can conclude that $(G/\{(c_j, x_{ij})\}, k - 1, d)$ is a \no\ instance.
This contradicts Observation~\ref{obs:reduced-inst-yes}.
Hence our assumption is wrong and $F \cap E(C, V(H))$ is an empty set. \end{proof}

\begin{claim} $F \cap E(R, V(H)) = \emptyset$.
\end{claim}
\begin{proof}
Assume that is $F \cap E(R, V(H))$ is not empty.
As any vertex $c_j$ in $C$ has degree $d + 1$, edges in $F$ merge at least two vertices in $N[c_j]$ (Lemma~\ref{lemma:disjoint-nbrs}).
%By Claim~$1$, $2$, and $3$, all edges in $F$ that merge vertices in $N[c_j]$ are either in $E(R, V(H))$, $E(V(H), V(H))$, or $E(V(H), S)$.
We argue that if our assumption is correct then there are not enough edges to merge two vertices in each $N[c_j]$.

Let $J'$ be the set of columns such that there is no edge of the form $(x_{ij'}, s)$ in $F$, where $i, j' \in [k]$ and $s \in S$.
Note that set $J'$ is not empty as there are $k$ columns and at most $(k - 1)$ edges in $F \cap E(V(H), S)$.
There are at most $|F| - (k - |J'|) = |J'|$ many edges to merge two vertices in $N[c_j']$ for each $j'$ in $J'$.
For any two different vertices $c_j, c_{j'}$ in $C$, their neighbourhoods outside $S$ do not intersect. Formally, $(N[c_j] \setminus S) \cap (N[c_{j'}] \setminus S) = \emptyset$.  
Hence, $|J'|$ many edges need to cover $2 |J'|$ vertices.
This implies that edges in $F \setminus E(V(H), S)$ form a matching in $G$. 
For any vertex $c_j$ in $C$, its neighborhood is an independent set.
Hence, the only possible way to merge two vertices in each $N[c_{j'}]$ using edges in matching is to contract an edge incident to $c_{j'}$ and one of its neighbors in $V(H)$.
Hence, all the edges in $F \setminus E(V(H), S)$ are in $E(C, V(H))$.
This is a contradiction to Claim~\ref{claim:IS-backward-second} which states that there is no solution edge in $E(C, V(H)$.
Hence our assumption was wrong and $F \cap E(R, V(H)) = \emptyset$. \end{proof}

\begin{claim}
\label{claim:IS-backword-end}
$F \cap E(V(H), V(H)) = \emptyset$.
\end{claim}
\begin{proof}
Assume that there exists an edge, say $(x_{ij}, x_{i'j'})$, in $F \cap E(V(H), V(H))$ for some $i, j, i', j' \in [k]$.
Consider instance $(G/\{(x_{ij}, x_{i'j'})\}, k - 1, d)$.
As any vertex $r_i$ in $R$ has degree $d + 1$, edges in any solution for the instance merge at least two vertices in $N[r_i]$ (Lemma~\ref{lemma:disjoint-nbrs}).
Hence, $F \setminus \{(x_{ij}, x_{i'j'})\}$ merges at least two vertices in $N[r_i]$ for each $r_i$ in $R$. 
Let $Y$ be the set of vertices in $G/\{(x_{ij}, x_{i'j'})\}$ such that $Y$ contains at least two vertices in $N[r_i]$ for every $r_i$ in $R$.
Note that the cardinality of $Y$ is at least $2k - 1$.
By Claim~$4$, there is no edge in $F \cap E(R, V(H)$.
This implies that $(k - 1)$ in $F \setminus \{(x_{ij}, x_{i'j'})\}$ covers at least $2k - 1$ vertices.
This is a contradiction as any edge can cover at most two vertices.
Hence our assumption is wrong and $F \cap E(V(H), V(H)) = \emptyset$. 
\end{proof}

\begin{proof}(of Lemma~\ref{lemma:correct-reduction-PermIS-MaxDegCont-backword})
By Claims~\ref{claim:IS-backword-start} to \ref{claim:IS-backword-end}, every edge in $F$ is of the form $(x_{ij}, s_l)$ for some $i, j, l \in [k]$ where $x_{ij} \in V(H)$ and $s_l \in S$.
Let $X'$ be the collection of vertices in $V(H)$ that are endpoints of some edges in $F$.
The size of $X'$ is at most $k$.
In the remaining part, we argue that $X'$ is an independent set in $H$ and contains one vertex from each row and column.
 
We first argue that for every vertex in $s_l$ in $S$, there is exactly one edge in $F$ which is incident to $s_l$.
Note that $S$ is an independent set and every vertex in it has the degree $d + 1$. 
By Lemma~\ref{lemma:disjoint-nbrs}, edges in $F$ merges at least two vertices in $N[s_l]$ for every $s_l$. 
As $|F| \le k$ and $|S| = k$, there is exactly one edge incident to $s_l$ for every $l \in [k]$. 

We now prove that $X'$ contains one vertex from every row and column. 
Recall that for every $i \in [k]$, the degree of vertex $r_i$ is $d + 1$ in $G$ and $N(r_i)$ contains $s_i$ and all the vertices in $i^{th}$ row.
By Lemma~\ref{lemma:disjoint-nbrs}, for every $i$, edges in $F$ merge at least two vertices in $N[r_i]$.
Hence, the other endpoint of the edge incident $s_i$ is some vertex in $i^{th}$ row.
This implies there exists a vertex in $X'$ from each row.
By similar arguments, we can prove that there exists a vertex in $X'$ from each column.

It remains to argue that $X'$ is an independent set in $H$. 
Define function $\phi : S \rightarrow V(H)$ as follows: For every $s_l \in S$, assign $\phi(s_l) = x_{ij}$ if $(x_{ij}, s_l) \in F$ for some $i, j \in [k]$.
As there is exactly one edge in $F$ which is incident to $s_l$, function $\phi$ is well defined.
Assume that there exists an edge $(\phi(s_l), \phi(s_{l'}))$ in $H$.
Let $s^{\circ}_l$ and $s^{\circ}_{l'}$ be the two new vertices added while contracting edges $(s_l, \phi(s_l))$ and $(s_{l'}, \phi(s_{l'}))$.
Note that $s^{\circ}_l$ is adjacent with $s^{\circ}_{l'}$, one vertex in $R$, $k$ vertices in $R C$, and $k^2 - 1$ vertices in $V(H)$.
Hence, degree of $s^{\circ}_l$ is $1 + 1 + k + k^2 - 1 = d + 1$.
This contradicts the fact that every vertex in $G/F$ has degree at most $d$. 
Hence our assumption was wrong and there is no edge $(\phi(s_l), \phi(s_{l'}))$ in $H$.
Since, $s_l, s_{l'}$ are any two arbitrary vertices in $S$, we can conclude that $X'$ is an independent set in $H$.

Hence, if $(G, k, d)$ is a \yes\ instance than so is $(H, k)$.
\end{proof}

We are now in a position to present a proof of Theorem~\ref{thm:brute-force-algo-lb} using Proposition~\ref{prop:ind-set-lower-bound}, Lemma~\ref{lemma:correct-reduction-PermIS-MaxDegCont-forward}, and Lemma~\ref{lemma:correct-reduction-PermIS-MaxDegCont-backword}.

\begin{proof}(of Theorem~\ref{thm:brute-force-algo-lb})
Assume, for the sake of contradiction, that there is an algorithm, say $\calA$, that given any instance $(G, k, d)$ of \textsc{MDC} runs in time $n^{o(k)}$ and correctly determines whether it is a \yes\ instance or not.
Using this algorithm as subroutine, we construct an algorithm to solve $(k \times k)$-\textsc{PIS}.

Consider Algorithm $\calB$ that given an instance $(H, k)$ of $(k \times k)$-\textsc{PIS} construct an instance $(G, k, d)$ of \textsc{MDC} as described in the reduction.
Then, it calls Algorithm $\calA$ as subroutine on instance $(G, k, d)$.
If Algorithm $\calA$ returns \yes\ then Algorithm $\calB$ returns \yes\ otherwise it returns \no.
The correctness of Algorithm $\calB$ follows from the correctness of Algorithm $\calA$, Lemma~\ref{lemma:correct-reduction-PermIS-MaxDegCont-forward}, and Lemma~\ref{lemma:correct-reduction-PermIS-MaxDegCont-backword}.
We now argue the running time of Algorithm $\calB$.
By the description of the reduction, it is easy to see that given instance $(H, k)$, the algorithm computes instance $(G, k, d)$ in time polynomial in $|V(H)| \in \calO(k^2)$ and $|V(G)| = n \in \calO(k^2)$.
Hence, the total running time of the algorithm is $n^{o(k)} = k^{o(k)}$.

This implies there is an algorithm to solve $(k\times k)$-\textsc{PIS} in time $k^{o(k)}$.
But this contradicts Proposition~\ref{prop:ind-set-lower-bound}.
Hence, our assumption is wrong which concludes the proof.
\end{proof}
\section{A Different \FPT\ Algorithm}
\label{sec:fpt-algo}
In this section, we present a different \FPT\ algorithm for \textsc{Maximum Degree Contraction}.
We introduce a variation of the problem called \textsc{Labeled-Maximum Degree Contraction (Labeled-MDC)}.
We present an \FPT\ algorithm for \textsc{Labeled-MDC} and use it as a subroutine to present an \FPT\ algorithm for \textsc{MDC}.

Informally, an instance of \textsc{Labeled-MDC} is an instance of \textsc{MDC} along with a labeling of vertices in the graph.
Every vertex has a red or blue label.
We are only interest in a solution that satisfies the following properties:
$(1)$ every edge has red labelled endpoints, and
$(2)$ for any red-labelled maximal connected component, a solution either spans none or all the vertices in that component.
We remark that because of the second condition, this problem is \emph{not} a restricted version of \textsc{MDC}.
We formally define \textsc{Labeled-MDC} as follows.

\defparproblem{\textsc{Labeled-MDC}}{Graph $G$, a partition $V_r, V_b$ of $V(G)$, and integers $k, d$}{$k + d$}{Does there exist a subset $F$ of $E(G)$ of size at most $k$ such that $(a)$ every vertex in $G/F$ has degree at most $d$; $(b)$ $V(F) \subseteq V_r$; and $(c)$ for a connected component $C$ of $G[V_r]$, if $C \cap V(F) \neq \emptyset$ then $C \subseteq V(F)$.
}

We say a set of edges $F$ is a \emph{solution} to instance $(G, (V_r, V_b), k, d)$ if the number of edges in $F$ is at most $k$, the maximum degree of graph $G/F$ is at most $d$, $V(F) \subseteq V_r$ and for a connected component $C$ of $G[V_r]$, if $C \cap V(F) \neq \emptyset$ then $C \subseteq V(F)$.

It is easy to see that if $(G, (V_r, V_b), k, d)$ is a \yes\ instance of \textsc{Labeled-MDC} then $(G, k, d)$ is a \yes\ instance of \textsc{MDC}.
Let $\calU$ be the family of all subsets of $V(G)$.
If $(G, k, d)$ is a \yes\ instance of \textsc{MDC} then $(G, (V_r, V(G) \setminus V_r), k, d)$ is a \yes\ instance of \textsc{Labeled-MDC} for some set $V_r$ in $\calU$.
We use \emph{universal sets} to construct a `small’ family of subsets of $V(G)$ that suffices for our purpose.
We assume that there is a unique integer in $[n]$ for every vertex in $V(G)$.
We use a subset of $[n]$ and a corresponding subset of $V(G)$ interchangeably. 
\begin{definition}[Universal Sets] An $(n, l)$-universal set is a family $\calU$ of subsets of $[n]$ such that for any $S \subseteq [n]$ of size $l$, the family $\{A \cap S\ |\ A \in \calU \}$ contains all subsets of $S$.
\end{definition}
\begin{proposition}[\cite{alon1995color}]\label{prop:universal-set} For any $n, l \ge 1$ one can construct an $(n, l)$-universal set of size $2^{\calO(l)} \cdot \log(n)$ in time $2^{\calO(l)} \cdot n \log(n)$.
\end{proposition}

In the following lemma, we argue that an \FPT\ algorithm for \textsc{Labeled-MDC} leads to an \FPT\ algorithm for \textsc{MDC}.
\begin{lemma}\label{lemma:fpt-col-to-MDC} Suppose there is an algorithm that given an instance $(G, (V_r, V_b), k, d)$ of \textsc{Labeled-MDC} runs in time $f(k, d) \cdot n^{\calO(1)}$ and correctly determines whether it is a \yes\ instance.
Then, there is an algorithm that given an instance $(G, k, d)$ of \textsc{MDC} runs in time $2^{\calO(dk)} \cdot f(k, d) \cdot n^{\calO(1)}$ and correctly determines whether it is a \yes\ instance.
\end{lemma}
\begin{proof}
Let $\calA$ be an algorithm that given an instance $(G, (V_r, V_b), k, d)$ of \textsc{Labeled-MDC} runs in time $f(k, d) \cdot n^{\calO(1)}$ and correctly determines whether it is a \yes\ instance.
We first describe an algorithm for \textsc{MDC} that uses $\calA$ as a subroutine.
For the input $(G, k, d)$, the algorithm constructs a $(U, 2k + kd)$-universal family $\calU$ using Proposition~\ref{prop:universal-set}.
For every set $V_r$ in $\calU$, the algorithm runs Algorithm~$\calA$ with input $(G, (V_r, V(G) \setminus V_r), k, d)$. 
The algorithm returns \yes\ if Algorithm~$\calA$ returns \yes\ for one of these inputs otherwise it returns \no.
This completes the description of the algorithm.
The running time of the algorithm follows from the description and Proposition~\ref{prop:universal-set}.
In the remaining proof, we argue the correctness of the algorithm. 
More precisely, we prove that $(G, k, d)$ is a \yes\ instance of \textsc{MDC} if and only if there is a subset $V_r$ in $\calU$ such that 
$(G, (V_r, V(G) \setminus V_r), k, d)$ is a \yes\ instance of \textsc{Labeled-MDC}.

Suppose that $(G, k, d)$ is a \yes\ instance of \textsc{MDC} and let $F$ be a solution to it.
Note that $|V(F)| \le 2k$.
We first argue that the number of vertices in $N(V(F))$ is at most $kd$.
Let $\calW$ be the $G/F$-witness structure of $G$ and $\psi : V(G) \rightarrow V(G/F)$ be the corresponding function.
Consider an arbitrary vertex $v$ in $N(V(F))$.
As $v$ is not in $V(F)$, $\psi(v)$ corresponds to a small witness set in $\calW$.
As $v$ is in $N(V(F))$, $\psi(v)$ is adjacent to a vertex in $G/F$ that corresponds to a big witness set in $\calW$.
By Observation~\ref{obs:witness-structure-property} (\ref{item:nr-big-witness-set}), there are at most $k$ big witness sets in $\calW$.
Since the maximum degree of $G/F$ is at most $d$, there are at most $kd$ small witness sets in $\calW$ that are adjacent with some big witness set.
Hence, there are at most $kd$ vertices in $N(V(F))$.
As $\calU$ is a $(n, 2k + dk)$-universal set and $|N[V(F)]| \le 2k + dk$, there exists a set $A$ in $\calU$ such that the family $\{A \cap N[V(F)]~|\ A \in \calU \}$ contains all subsets of $N[V(F)]$.
This implies, there exists a set, say $V_r$, such that $V_r \cap N[V(F)] = V(F)$.
We argued that $(G, (V_r, V(G) \setminus V_r), k, d)$ is a \yes\ instance of \textsc{Labeled-MDC}.

Note that $G/F$ has maximum degree at most $d$ and $V(F) \subseteq V_r$.
We need to prove that for a connected component $C$ of $G[V_r]$ if $C \cap V(F) \neq \emptyset$ then $C \subseteq V(F)$. 
Assume that there exits a connected component $C$ of $G[V_r]$ such that $C \cap V(F) \neq \emptyset$ and $C \setminus V(F) \neq \emptyset$.
As $C$ is a connected component and $C \cap V(F) \neq \emptyset$, there exists a vertex $v$ in $C \setminus V(F)$ that is adjacent with some vertex in $V(F)$.
Hence, there is a vertex in $N(V(F)) \cap V_r$.
This contradicts the fact that $V_r \cap N[V(F)] = V(F)$.
Hence, our assumption is wrong and $C \setminus V(F)$ is an empty set.
This implies $(G, (V_r, V(G)\setminus V_r), k, d)$ is a \yes\ instance of \textsc{Labeled-MDC}. 
As mentioned before, it is easy to see that if $(G, (V_r, V_b), k, d)$ is a \yes\ instance of \textsc{Labeled-MDC} then $(G, k, d)$ is a \yes\ instance of \textsc{MDC}. This concludes the proof of the lemma.
\end{proof}

In the remaining section, we present a recursive algorithm for \textsc{Labeled-MDC}.
We start with the following simple reduction rules.
\begin{reduction rule}\label{rr:trivial-yes} For an instance $(G, (V_r, V_b), k, d)$, if the maximum degree of vertices in $G$ is at most $d$ and $k \ge 0$ then return a \yes\ instance.
\end{reduction rule}

It is easy to see that the first reduction rule is safe.
Recall that a set of edges $F$ is called solution to $(G, (V_r, V_b), k, d)$ if the number of edges in $F$ is at most $k$, the maximum degree of graph $G/F$ is at most $d$, $V(F) \subseteq V_r$, and for a connected component $C$ of $G[V_r]$, if $C \cap V(F) \neq \emptyset$ then $C \subseteq V(F)$.
Consider a connected component $C$ of $G[V_r]$.
If $|C| = 1$ then no solution edge can be incident to it. 
Also, if $|C| \ge 2k + 1$ then because of the last property and the fact that $|V(F)| \le 2k$, no solution edge can be incident to vertices in $C$.
These simple observations prove that the following reduction rule is safe.
\begin{reduction rule}\label{rr:size-red-comp} 
For an instance $(G, (V_r, V_b), k, d)$, if there is a connected component, say $C$, of $G[V_r]$ such that $|C| = 1$ or $|C| \ge 2k + 1$ then move $C$ from $V_r$ to $V_b$ i.e. return instance $(G, (V_r \setminus C, V_b \cup C), k, d)$.
\end{reduction rule}

By Observation~\ref{obs:high-deg-vertex}, vertex $v$ in $V_b$ can be adjacent to at most $d + k$ vertices in $V_r$. 
The following reduction rule ensures that the neighbors of $v$ in $V_r$ are not spread across many connected components. 
\begin{reduction rule}\label{rr:blue-nbr} For an instance $(G, (V_r, V_b), k, d)$, if there exists a vertex, say $v$, in $V_b$ for which $N_G(v)$ intersects with $d + 1$ different connected components of $G[V_r]$ then return a \NO\ instance.
\end{reduction rule}
\begin{lemma} \label{lemma:rr-blue-nbr} Reduction Rule~\ref{rr:blue-nbr} is safe.
\end{lemma}
\begin{proof}
Assume that $(G, (V_r, V_b), k, d)$ is a \yes\ instance.
Let $F$ be its solution and it contracts $G$ to $G/F$ via mapping $\psi$.
Suppose $C_1, C_2, \dots, C_{d+1}$ are connected components of $G[V_r]$ such that $C_i \cap N(v) \neq \emptyset$ for $i \in [d + 1]$.
For every $i$, consider a vertex, say $u_i$, in $C_i \cap N(v)$.
Let $U = \{u_1, u_2, \dots, u_{d + 1}\}$.
Define $\psi(U) = \bigcup_{u \in U} \psi(u)$.
For $i, j \in [d + 1]$, $i \neq j$ implies $\psi(u_i) \neq \psi(u_j)$ as $C_i$ and $C_j$ are two different connected components of $G[V_r]$ and $V(F) \subseteq V_r$.
This implies $|\psi(U)| = |U| = d + 1$.
As $V(F) \subseteq V_r$ and $v \in V_b$, $F$ does not contain an edge incident to $v$.
Hence, $\psi(v) \neq \psi(u_i)$ for any $i \in [d + 1]$.
As $\psi(U) \subseteq N_{G/F}(\psi(v))$ and $|\psi(U)| \ge d + 1$, vertex $\psi(v)$ is adjacent with $d + 1$ or more vertices in $G/F$.
This contradicts the fact that the maximum degree of vertices in $G/F$ is at most $d$.
Hence, our assumption was wrong and $(G, (V_r, V_b), k, d)$ is a \no\ instance.
\end{proof}

The algorithm exhaustively applies the reduction rules mentioned above.
On a reduced instance, the algorithm creates multiple instances using the following subroutine.
For an instance $(G, (V_r, V_b), k, d)$, a subset $R$ of $V_r$, and a $(d + 1)$-coloring of $R$, the subroutine creates a new instance by contracting each colored component of $R$ into a single vertex, and (re-)label it blue. 
We need the notion of `valid coloring’ to filter out colorings that will not produce a `smaller’ instance.
For graph $H$, a vertex coloring $\phi : V(H) \rightarrow [d + 1]$ is said to be a \emph{valid coloring} if every monochromatic connected component is of size at least two.
We now describe the subroutine.

\paragraph*{Subroutine \texttt{Colorwise-Contraction}} This subroutine takes as an input an instance $(G, (V_r, V_b), k, d)$ of \textsc{Labeled-MDC}, a non-empty subset $R$ of $V_r$, and a valid coloring $\phi$ of $G[R]$.
It returns another instance of \textsc{Labeled-MDC}. 
It initializes $G’ = G$, $V_r’ = V_r$, $V_b’ = V_b$, and $k’ = k$.
For a monochromatic connected component $C$ of $G[R]$, the subroutine finds a spanning tree of $G[C]$ and contracts all edges in it.
Let $v_C$ be the vertex obtained at the end of this series of edge contractions.
It updates $V’_r = V_r \setminus C$, $V’_b = V_b \cup \{v_C\}$ and reduces $k$ by $|C| - 1$.
The subroutine repeats this procedure for every monochromatic connected component of $G[R]$.
It returns $(G’, (V’_r, V’_b), k’, d)$ as instance of \textsc{Labeled-MDC}.
This completes the description of the subroutine.

It is easy to verify that $(V_r’, V’_b)$ is a partition of $V(G’)$.
As $\phi$ is a valid coloring of $G[R]$, a union of spanning trees of all monochromatic connected components of $G[R]$ contains at least $|R|/2$ edges.
Hence, the subroutine contracts at least $|R|/2$ edges.
This small observation will be helpful to get a bound on the running time of the algorithm. 
\begin{remark}
\label{remark:reduced-graph-prop}
$k’ \le k - |R|/2$.
\end{remark}

Let $\texttt{CC}[(G, (V_r, V_b), k, d); R; \phi]$ denote the instance returned by the subroutine when the input is $(G, (V_r, V_b), k, d)$, $R$, and $\phi$.
In the following lemma, we prove if the original instance is a \yes\ instance than at least one of the reduced instances is a \yes\ instance. 
\begin{lemma}
\label{lemma:colorwise-contraction-correct}
Consider a \yes\ instance $(G, (V_r, V_b), k, d)$ of \textsc{Labeled-MDC}.
Let $R$ be a union of some connected components of $G[V_r]$.
Suppose there is solution $F$ to $(G, (V_r, V_b), k, d)$ such that $R \subseteq V(F)$.
Then, there is a valid coloring $\phi : R \rightarrow [d + 1]$ of $G[R]$ for which $\emph{\texttt{CC}}[(G, (V_r, V_b), k, d); R; \phi]$ is a \yes\ instance.
\end{lemma}
\begin{proof}
Let $H = G/F$.
Consider the $H$-witness structure $\calW$ of $G$ and let $G$ be contracted to $H$ via $\psi$.
Define a subset $\calW_R$ of $\calW$ as the collection of witness sets that intersects $R$.
Formally, $\calW_R = \{W \in \calW\ |\ W \cap R \neq \emptyset \}$.
Let $\calW_R = \{W_1, W_2, \dots, W_q\}$.
For every $i \in [q]$, let $h_i$ be the vertex corresponding to $W_i$. 
In other words, $W_i = \{v \in V(G)\ |\ \psi(v) = h_i\}$.
Let $R_H = \{h_1, h_2, \dots, h_q\}$.

Let $F_1$ be the collection of edges in $F$ that are incident to some vertex in $R$.
Hence, $R \subseteq V(F_1)$.
As $R$ is the union of connected components in $G[V_r]$ and $V(F_1) \subseteq V(F) \subseteq V_r$, we can conclude that $R = V(F_1) = \bigcup_{i \in [q]} W_i$.
Hence, $\{W_1, W_2, \dots, W_q\}$ is a partition of $R$.
As there is a solution edge incident to every vertex in $R$, every witness set in $\calW_R$ is a big witness set. 
This implies for every $i \in [q]$, there is a subset $F_i$ of $F$ such that $W_i = V(F_i)$.
As the maximum degree of vertices in graph $H$ is at most $d$, there is a proper $(d+1)$-coloring, say $\gamma$, of $H$.
For $i, j \in [q]$, if $(h_i, h_j)$ is an edge in $H$ then $\gamma(h_i) \neq \gamma(h_j)$.
Define a coloring $\phi:R \rightarrow [d + 1]$ as follows.
For $v \in R$, $\phi(v) = \gamma(h_i)$ where $v \in W_i$.
As $\{W_1, W_2, \dots, W_q\}$ is a partition of $R$, function $\phi$ is well defined.
Since $W_i$ is a big witness set, $\phi$ is a valid coloring.

By the construction of $\phi$, any witness set in $\calW$ is monochromatic.
Since $\gamma$ is a proper coloring of $H$, any two witness sets adjacent to each other have distinct colors.
Hence, every witness set in $\calW_R$ is a monochromatic connected component of coloring $\phi$.
As algorithm constructs every valid coloring of $R$, it also consider this coloring and create instance $(G’, (V_r’, V_b’), k’, d) = \texttt{CC}[(G, (V_r, V_b), k, d), R; \phi]$.
For every $i \in [q]$, let $F^{\circ}_i$ be edges in a spanning tree of $G[W_i]$.
Define $F^{\circ} = \bigcup_{i \in [q]} F^{\circ}_i$.
As $W_i = V(F_i)$, graphs $G/F_1$ and $G/F^{\circ}$ are identical.
Also, $|F^{\circ}_i| \le |F_i|$ which implies $|F^{\circ}| \le |F_1|$.
Define $F^{\star} = (F \setminus F_1) \cup F^{\circ}$.
It is easy to verify that $F^{\star}$ is also a solution to $(G, (V_r, V_b), k, d)$.

We now argue that $F^{\star} \setminus F^{\circ}$ is a solution to $(G’, (V_r’, V_b’), k’, d)$.
By the description of the algorithm, $k’ = k - |F^{\circ}|$.
As $|F^{\star}| \le |F| \le k$ and $F^{\circ} \subseteq F^{\star}$, we have $|F^{\star} \setminus F^{\circ}| \le |F^{\star}| - |F^{\circ}| \le k’$.
Note that $G/F^{\star} = (G/F^{\circ})/(F^{\star}\setminus F^{\circ}) = G’/(F^{\star}\setminus F^{\circ})$ as $G’ = G/F^{\circ}$. 
This implies the maximum degree of $G’/(F^{\star}\setminus F^{\circ})$ is at most $d$.
The only thing that remains to argue is that $V(F^{\star}\setminus F^{\circ})$ is contained in $V_r’$.
By construction, $F \setminus F_1 = F^{\star} \setminus F^{\circ}$.
As $F_1$ is the set of edges in $F$ that were incident to $R$, we can conclude that no edge in $F^{\star} \setminus F^{\circ}$ is incident to $R$.
Recall that $V_r’ = V_r \setminus R$.
Hence, $V(F^{\star} \setminus F^{\circ}) \subseteq V_r’$.
This implies that $F^{\star} \setminus F^{\circ}$ is a solution to $(G’, (V_r’, V_b), k’, d)$ and concludes the proof of the lemma.
\end{proof}

In the above lemma, instead of considering any arbitrary subset $V_r$ we only consider a subset that is a union of one or more connected components of $G[V_r]$.
This suffices for our purpose as the algorithm calls the subroutine only on such subsets of $V_r$.
Also, note that we do not need to know the solution $F$ explicitly to apply the above lemma.
It suffices to know that such a solution exists. 
We are now able to present an algorithm for \textsc{Labeled-MDC}

\paragraph*{Algorithm for \textsc{Labeled-MDC}}
The algorithm takes as input an instance $(G, (V_r, V_b), k, d)$ of \textsc{Labeled-MDC} and returns \yes\ or \no.
If $k < 0$ then the algorithm returns \no.
If $k = 0$ then it finds the maximum degree of $G$.
If it is at most $d$ then the algorithm returns \yes\ otherwise it returns \no.
The algorithm exhaustively applies Reduction Rules~\ref{rr:trivial-yes}, \ref{rr:size-red-comp}, and \ref{rr:blue-nbr}.
If the reduced instance is a trivial \yes\ (resp. \no) instance then the algorithm returns \yes\ (resp. \no).
Otherwise, it creates multiple instances and makes recursive calls on these instances.
The algorithm returns \yes\ if one of the recursive calls returns \yes, otherwise; it returns \no.

We now describe the procedure used by the algorithm to create new instances.
Let $(G, (V_r, V_b), k, d)$ be the instance on which reduction rules are not applicable.
The algorithm finds a vertex, say $v$, in $G$ such that $\deg_G(v) \ge d + 1$. 
It considers the following two cases.
\begin{enumerate}
\item (Vertex $v$ is in $V_r$) 
Let $R$ be the connected component of $G[V_r]$ that contains $v$.
The algorithm constructs all valid coloring $\phi : R \rightarrow [d + 1]$ of $G[R]$.
For each coloring, the algorithm calls subroutine \texttt{Colorwise-Contraction} with input $(G, (V_r, V_b), k, d)$, $R$, and $\phi$.
The algorithm calls itself with the instances returned by this subroutine as the input.

\item (Vertex $v$ is in $V_b$) 
Let $C_1, C_2, \dots, C_q$ be the connected components of $G[V_r]$ such that $N(v) \cap C_i \neq \emptyset$ for every $i \in [q]$.
For a non-empty subset $I \subseteq [q]$, define $R_I := \bigcup_{i \in I} C_i$.
For every non-empty subset $I \subseteq [q]$, the algorithm proceeds as follows.
If $|R_I| \ge 2k + 1$, the algorithm discards this choice of $I$ and moves to the next one.
Otherwise, the algorithm constructs all valid colorings $\phi : R_I \rightarrow [d + 1]$ of $G[R_I]$.
For each coloring, the algorithm calls subroutine \texttt{Colorwise-Contraction} with input $(G, (V_r, V_b), k, d)$, $R_I$, and $\phi$.
The algorithm calls itself with the instance returned by this subroutine as input.
\end{enumerate}
This completes the description of the algorithm.

In the following lemma, we prove that the algorithm described above is correct and runs in the desired time.
\begin{lemma} \label{lemma:fpt-labeled-MDC}
There is an algorithm that given an instance $(G, (V_r, V_b), k, d)$ of \textsc{Labeled-MDC} runs in time $2^{(d + 2) k} \cdot (d + 1)^{2k} \cdot n^{\calO(1)}$ and correctly determines whether it is a \yes\ instance.
\end{lemma}
\begin{proof}
We argue that the algorithm described above solves \textsc{Labeled-MDC} in the desired time.
We prove this lemma by the induction on the solution size $k$.

Consider the base case when the solution size is zero.
Here, the algorithm finds a maximum degree of the graph and depending on its value returns \yes\ or \no. 
It is easy to see that the lemma holds in this case. 
Assume that the lemma is true when the solution size is at most $k - 1$.

We first prove that given a \yes\ instance the algorithm returns \yes.
Suppose $(G, (V_r, V_b), k, d)$ is a \yes\ instance of \textsc{Labeled-MDC} and let $F$ be its solution. 
Note that this implies that $F$ is a solution to $(G, k, d)$. 
If the algorithm returned \yes\ because Reduction Rule~\ref{rr:trivial-yes} returned a \yes\ instance then the lemma is vacuously true.
By Lemma~\ref{lemma:rr-blue-nbr}, Reduction Rule~\ref{rr:blue-nbr} is not applicable on the input.
Consider the instance obtained by the exhaustive application Reduction Rules~\ref{rr:trivial-yes} and \ref{rr:size-red-comp} on the input instance.
For notational convenience, we denote this reduced instance by $(G, (V_r, V_b), k, d)$.
As Reduction Rule~\ref{rr:trivial-yes} is not applicable, there is a vertex in $G$ that has degree at least $d + 1$.
Let $v$ be the vertex of degree at least $d + 1$ found by the algorithm. 
By Observation~\ref{obs:high-deg-vertex}, $V(F)$ intersects with $N[v]$.

Consider the case when $v$ is in $V_r$ and let $R$ be the connected component of $G[V_r]$ that contains $v$.
Since $V(F) \subseteq V_r$, we have $R \cap V(F) \neq \emptyset$.
As $F$ is a solution to $(G, (V_r, V_b), k, d)$, $R \cap V(F) \neq \emptyset$ implies $R \subseteq V(F)$.
Instance $(G, (V_r, V_b), k, d)$, subset $R$ of $V_r$, and solution $F$ satisfies the premise of Lemma~\ref{lemma:colorwise-contraction-correct}.
Hence, there is a valid coloring $\phi:R \rightarrow [d + 1]$ of $G[R]$ such that 
$\texttt{CC}[(G, (V_r, V_b), k, d), R; \phi]$ is a \yes\ instance.
As $R \neq \emptyset$, Remark~\ref{remark:reduced-graph-prop} implies that $k’ < k$. 
By the induction hypothesis, the algorithm correctly returns \yes\ when the input is $(G’, (V’_r, V’_b), k’, d)$.
As one of the recursive calls returns \yes, the algorithm returns \yes\ when the input is $(G, (V_r, V_b), k, d)$ and $v$ is in $V_r$.

Consider the case when $v$ is in $V_b$.
Let $C_1, C_2, \dots, C_q$ be the connected components of $G[V_r]$ such that $N(v) \cap C_i \neq \emptyset$ for every $i \in [q]$.
Recall that for a non-empty subset $I \subseteq [q]$, $R_I = \bigcup_{i \in I} C_i$.
As $V(F)$ intersects $N[v]$ and $V(F) \subseteq V_r$, there exists a non-empty subset $I’ \subseteq [q]$ such that for $i \in [q]$, $C_i \cap N(v) \neq \emptyset$ if and only if $i \in I’$.
As $F$ is a solution to $(G, (V_r, V_b), k, d)$, $C_i \cap V(F) \neq \emptyset$ implies $C_i \subseteq V(F)$.
Hence, $R_{I’} \subseteq V(F)$.
As $|V(F)| \le 2k$, $|R_{I’}| \le 2k$.
For every non-empty subset $I \subseteq [q]$ for which $|R_I| \le 2k$, the algorithm constructs all valid coloring $\phi : R_I \rightarrow [d + 1]$ of $G[R_I]$ and calls \texttt{Colorwise-Contraction}.
Instance $(G, (V_r, V_b), k, d)$, subset $R_{I’}$ of $V_r$, and solution $F$ satisfies the premise of Lemma~\ref{lemma:colorwise-contraction-correct}.
Hence, there is a valid coloring $\phi:R_{I’} \rightarrow [d + 1]$ of $G[R_{I’}]$ such that $ (G’, (V’_r, V’_b), k’, d) = \texttt{CC}[(G, (V_r, V_b), k, d), R_{I’}, \phi]$ is a \yes\ instance.
As $R \neq \emptyset$, Remark~\ref{remark:reduced-graph-prop} implies that $k’ < k$. 
By the induction hypothesis, the algorithm correctly returns \yes\ when the input is $(G’, (V’_r, V’_b), k’, d)$.
As one of the recursive calls returns \yes, the algorithm returns \yes\ when the input is $(G, (V_r, V_b), k, d)$ and $v$ is in $V_b$.
This implies that if $(G, (V_r, V_b), k, d)$ is a \yes\ instance then the algorithm returns \yes. 

We now prove that if the algorithm returns \yes\ on instance $(G, (V_r, V_b), k, d)$ then it is a \yes\ instance of \textsc{Labeled-MDC}.
If the algorithm returned \yes\ because Reduction Rule~\ref{rr:trivial-yes} returned a \yes\ instance then the lemma is vacuously true.
Otherwise, there is a newly created instance, say $(G’, (V’_r, V’_b), k’, d)$, on which the recursive call of the algorithm returned \yes.
Let $R$ be the subset of $V_r$ and $\phi$ be its valid coloring such that \texttt{Colorwise-Contraction} returned this instance when input was $(G, (V_r, V_b), k, d)$, $R$, and $\phi$. 
Let $F^{\circ}$ be the edges in $G$ contracted by the subroutine to contract $G’$.
In other words, $F^{\circ}$ is a collection of spanning trees of connected monochromatic components of $G[R]$.
Note that $|F^{\circ}| = k - k’$.
The algorithm calls \texttt{Colorwise-Contraction} only on non-empty subsets $R$.
Hence, by Remark~\ref{remark:reduced-graph-prop}, $k’ < k$. 
By the induction hypothesis, $(G’, (V’_r, V’_b), k’, d)$ is a \yes\ instance of \textsc{Labeled-MDC}.
It is easy to see that if $F’$ is a solution to $(G’, (V’_r, V’_b), k’, d)$ then $F’ \cup F^{\circ}$ is a solution to $(G, (V_r, V_b), k, d)$.
This concludes the proof of the correctness of the algorithm.

We now bound the running time of the algorithm.
The algorithm can apply all the reduction rules in polynomial time. 
It creates new instances only when none of the reduction rules are applicable. 
As Reduction Rules~\ref{rr:size-red-comp} is not applicable, any connected component of $G[V_r]$ has at least two and at most $2k$ vertices. 
In Case(1), the algorithm creates at most $(d + 1)^{|R|}$ many instances. 
By Remark~\ref{remark:reduced-graph-prop} and the induction hypothesis, the time taken by the algorithm in this case is
 
$$(d + 1)^{|R|} \cdot 2^{(d + 2)(k - |R|/2)} \cdot (d + 1)^{2(k - |R|/2)} \cdot n^{\calO(1)} \le 2^{(d + 2) k} \cdot (d + 1)^{2k} \cdot n^{\calO(1)}. $$

As Reduction Rule~\ref{rr:blue-nbr} is not applicable, for any vertex $v$ in $V_b$, there are at most $d$ connected components of $G[V_r]$ that intersects $N(v)$.
In Case(2), the algorithm constructs all valid partitions of $R_I$ only when $|R_I| \le 2k$.
Hence, in this case, the algorithm creates $2^d \cdot (d + 1)^{|R|}$ many instances.
By Remark~\ref{remark:reduced-graph-prop} and the induction hypothesis, the time taken by the algorithm in this case is 

$$2^d \cdot (d + 1)^{|R|} \cdot 2^{(d + 2)(k - |R|/2)} \cdot (d + 1)^{2(k - |R|/2)} \cdot n^{\calO(1)} \le 2^{(d + 2) k} \cdot (d + 1)^{2k} \cdot n^{\calO(1)}. $$

As $|R| \ge 2$, we have $2^d \cdot 2^{(d + 2)(-|R|/2)} \le 1$.
This completes the proof of the lemma.
\end{proof}

The correctness of Theorem~\ref{thm:fpt-new} immediately follows from Lemma~\ref{lemma:fpt-col-to-MDC} and Lemma~\ref{lemma:fpt-labeled-MDC}.
\section{No Polynomial Kernel}
\label{sec:no-poly-kernel}
In this section, we prove that \textsc{Maximum Degree Contraction} does not admit a polynomial kernel when parameterized by $k + d$.
To show that, we present a reduction from \textsc{Red Blue Dominating Set (RBDS)}.
In this problem, an input is comprised of a bipartite graph $H$ with a bipartition $(R, B)$ of $V(H)$, and a positive integer $l$. 
The question is, does there exist a subset $R’$ of $R$ of size at most $l$ such that $N(R’) = B$?
Without loss of generality, we can assume that $l + 3 < |B|$ and no vertex in $R$ is adjacent to all but one vertices in $B$.
We know the following result about the compression of the problem.
See, for example, Theorem~$15.18$ in \cite{saurabh-book}.

\begin{proposition}\label{prop:no-poly-kernel-RBDS} Unless $\NP \subseteq \coNP/poly$, \textsc{RBDS}, parameterized by $|B|$, does not admit a polynomial compression.
\end{proposition}

If $|R| > 2^{|B|}$ then there are at least two different vertices, say $r_1, r_2$ such that $N(r_1) = N(r_2)$.
It is easy to see that it is safe to delete one of these two vertices. 
In this case, we can ensure, in polynomial time, that $|R| \le 2^{|B|}$ by repeating the above process.
This implies $\log_2 |R| \le |B|$.
Hence, we get the following corollary of Proposition~\ref{prop:no-poly-kernel-RBDS}.
\begin{corollary}\label{corr:no-poly-kernel-RBDS} Unless $\NP \subseteq \coNP/poly$, \textsc{RBDS}, parameterized by $|B| + \log_2|R|$, does not admit a polynomial compression.
\end{corollary}

For the sake of clarity, we use both $|B|$ and $\log_2|R|$ as parameters instead of replacing $\log_2|R|$ by the larger parameter $|B|$.
For notational convenience, we assume that $\log_2|R|$ is an integer.
If this is not the case, one can add some isolated vertices in $R$ to ensure that $\log_2|R|$ is an integer.
This results in at most doubling of the number of vertices in it.

\begin{figure}[t]
\begin{center}
\includegraphics[scale=0.6]{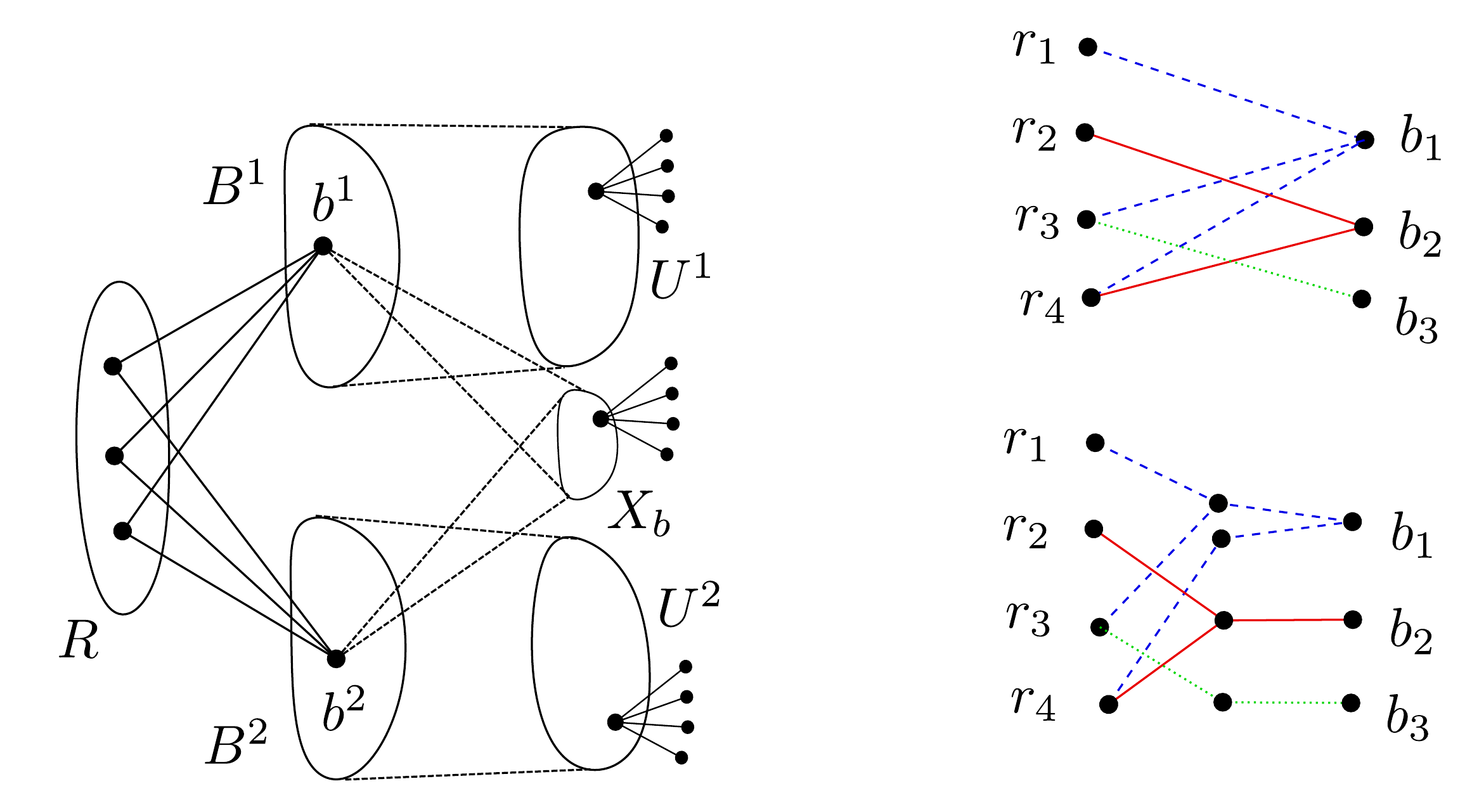}
\end{center}
\caption{(Left) Overview of the reduction. The doted lines indicate that there is a complete bipartite graph across two sets.
(Right) The operation of replacing edges incident to vertex in $B$ by a tree rooted at that vertex.
 \label{fig:RBDS-to-MDC-Overview}}
\end{figure}

We first present an overview of the reduction. 
Consider an instance $(H, R, B, l)$ of \textsc{RBDS}.
See Figure~\ref{fig:RBDS-to-MDC-Overview} for an illustration.
The reduction makes a copy of $R$ and two copies of $B$, say $B^1, B^2$.
For every vertex $b$ in $B$, we denote its two copies in $B^1, B^2$ by $b^1, b^2$, respectively.
For every edge $(r, b)$, the reduction adds edges $(r, b^1)$ and $(r, b^2)$.
It adds two independent sets $U^1, U^2$.
For every vertex $u \in U^1 \cup U^2$, it adds some pendent vertices adjacent to it. 
The reduction adds all edges to make a complete bipartite graph with $(B^1, U^1)$ as its bipartition.
Similarly, it adds all edges to make a complete bipartite graph with $(B^2, U^2)$ as its bipartition.
For every vertex $b$ in $B$, it adds a set of independent vertices $X_b$.
For every $x$ in $X_b$, it adds some pendent vertices adjacent to it and adds edges $(b^1, x)$, $(b^2, x)$.
We briefly present an intuition behind the construction before presenting the last step.
Let $G$ be the graph constructed so far and $k, d$ be two integers whose values depend only on $|B|, \log_2|R|$.
Suppose the reduction returns $(G, k, d)$ as an instance of \textsc{MDC}.

We set the value of $d$ and the number of pendant vertices such that it is ensured that the only vertices in $U^1 \cup U^2 \cup X_b$ have degree more than $d$ in $G$.
We fix $k$ and the sizes of sets $U^1, U^2, X_b$ to ensure that any solution for the reduced instance of \textsc{MDC} satisfy the following properties.
\begin{enumerate}
\item \label{item:no-BU} It does not include an edge with one of its endpoints in $B^1 \cup B^2$ and another in $U^1 \cup U^2$.
\item \label{item:no-BX} For any $b$ in $B$, it does not include an edge with one of its endpoints in $\{b^1, b^2\}$ and another in $X_b$. 
\item \label{item:no-pendent} It spans all vertices in $B^1 \cup B^2$. In other words, $B^1 \cup B^2 \subseteq V(F)$.
\item \label{item:l-parts} There are at most $l$ witness sets in the $G/F$-witness structure of $G$ that contain vertices in $B^1$ (similarly in $B^2$).
\item \label{item:twins-same} For every $b$ in $B$, $F$ includes $b^1, b^2$ in the same witness set.
\end{enumerate}

Property~(\ref{item:l-parts}) ensures that the degree constraints for the vertices in $U^1$ (similarly in $U^2$) are satisfied. 
Property~(\ref{item:twins-same}) ensures that for every $b$ in $B$, the degree constraints for the vertices in $X_b$ are satisfied.
Because of Property~(\ref{item:no-BU}) and (\ref{item:no-BX}), only the vertices in $R$ can make a witness set connected.
Hence, each witness set should contain at least one vertex from $R$.
We set the budget $k$ such that each witness set contains exactly one vertex from $R$.
To prove connectivity to witness set, this vertex needs to be adjacent to all vertices in that witness set.
Hence, the set of endpoints of edges in a solution to $(G, k, d)$ contains at most $l$ vertices in $R$ that dominates $B$.
This naturally leads to a solution to $(H. R, B, l)$.

We now present the last step in the construction.
The degree of the vertices formed by contracting a witness set can be larger than $d$.
To avoid this, we replace edges across $R, B$ that are incident to vertex $b$ in $B$ by a binary tree rooted at that vertex.
We ensure that for every edge incident $b$, there is a unique root-to-leaf path in the binary tree rooted at $b$ and vice versa. 

\paragraph*{Reduction}
Given an instance $(H, R, B, l)$ of \textsc{RBDS} as an input, the reduction outputs an instance $(G, k, d)$ of \textsc{MDC}.
The reduction creates an intermediate instance $(H^{\star}, R^{\star}, B^{\star}, l)$ of \textsc{RBDS}.
It takes a copy of $R$ and two copies of $B$, namely $B^1, B^2$, to create vertex set of graph $H^{\star}$.
Formally, $R^{\star} = R$ and $B^{\star} = B^1 \cup B^2$.
For every edge $(r, b)$ in $H$ such that $r \in R$ and $b \in B$, it adds edges $(r, b_1)$, $(r, b_2)$ to $H^{\star}$.
Here, vertices $b^1, b^2$ are copies of $b$ in $B^1, B^2$, respectively.
It is easy to see that $(H, R, B, l)$ is a \yes\ instance of \textsc{RBDS} if and only if $(H^{\star}, R, B^1 \cup B^2, l)$ is a \yes\ instance.

The reduction sets $k = 2|B| \cdot \log_2|R|$ and $d = 2|B| \cdot (\log_2|R| + k + 2)$, and constructs graph $G$ by modifying a copy of graph $H^{\circ}$ in the following way. 
It repeats the first three steps for $i = 1$ and $i = 2$.

\begin{itemize}
\item[-] For every vertex $b^i$ in $B^i$, it deletes all the edges incident to $b^i$ and constructs a binary tree that satisfies the following conditions:
$(i)$ the tree is rooted at $b^i$, 
$(ii)$ the height of the binary tree is $\log_2|R|$, and
$(ii)$ its leaves are the vertices in $N_{H^{\star}}(b^i)$.
Note that every edge incident to $b^{i}$ in $H^{\star}$ corresponds to an unique root-to-leaf path in this binary tree and vice versa.
Let $I^{i}$ be the collection of all the new vertices added in this step.
It is the collection of vertices in binary trees that are not roots or leaves. 
\item[-] It adds a set $U^i$ of $k + 1$ new vertices to $V(G)$. For every vertex $b^i$ in $B^i$ and every vertex $u$ in $U^i$, it adds edge $(b^i, u)$. It adds all edges to make $G[B^i \cup U^i]$ a complete bipartite graph with $B^i, U^i$ as its bipartition. 
\item[-] For every vertex $u$ in $U^i$, it adds $d - l$ pendant vertices adjacent to $u$.
\item[-] For every vertex $b$ in $B$, it adds a set $X_b$ of $k + 1$ new vertices. 
For every vertex $x$ in $X_b$, it adds edges $(b^1, x)$ and $(b^2, x)$. 
It also adds $d - 2$ pendant vertices adjacent to $x$.
Let $P_X^b$ be the set of pendent vertices adjacent to some vertex in $X_b$.
\end{itemize}
This completes the construction of $(G, k, d)$.

The following lemma identifies the set of vertices in $G$ that has degree more than $d$.
\begin{lemma}\label{lemma:RBDS-Red-Prop} 
Suppose the reduction returns $(G, k, d)$ when the input is $(H, R, B, l)$.
Then, $U \cup X$ is the collection of all vertices in $G$ that has degree strictly greater than $d$.
Here, $U = U^1 \cup U^2$ and $X = \bigcup_{b \in B} X_b$.
\end{lemma}
\begin{proof}
Define $I := I^1 \cup I^2$ and $P := P^1 \cup P^2 \cup \left(\bigcup_{b \in B} P^b_X \right)$.
Note that sets $R, I, B^1, B^2, U, X, P$ forms a partition of $V(G)$.
Consider a vertex $r$ in $R$. 
This vertex is a leaf in the binary trees rooted at every vertex in $N_{H^{\star}}(r)$.
Hence, in graph $G$, any vertex in $R$ is adjacent to at most $2|B| = |B^1 \cup B^2|$ many vertices in $I$. 
Every vertex in $I$ is an internal vertex in a binary tree and hence adjacent to at most three vertices. 
Every vertex in $B^1$ is adjacent to at most two vertices from set $I$, all the vertices in set $U^1$ and all the vertices in $X_b$.
Hence, vertex in $B^1$ is adjacent with $2 + |U^1| + |X_b|$ many vertices.
As $|U^1| = |X_b| = k + 1$, every vertex in $B^1$ is adjacent to at most $2k + 4 \le d$ vertices. 
By similar arguments, every vertex in $B^2$ is adjacent to at most $d$ vertices.
As $P$ is a collection of pendant vertices, every vertex in it is adjacent with exactly one vertex.

We have proved that every vertex in $V(G) \setminus (U \cup X)$ has degree at most $d$.
It remains to prove that every vertex in $U \cup X$ has degree at least $d$.
Every vertex in $U^1$ is adjacent with $d - l$ pendant vertices and every vertex in $B^1$.
As $|B| > l$, every vertex in $U^1$ has degree strictly greater than $d$.
By similar arguments, every vertex in $U^2$ has degree strictly greater than $d$.
For every $b$ in $B$, every vertex in $X_b$ is adjacent with $d - 1$ pendant vertices and two vertices in $B^1 \cup B^2$.
Hence, every vertex in $X$ is adjacent with at least $d + 1$ vertices.
This concludes the proof of the lemma.
\end{proof}

In the remaining section, we argue that the reduction is correct.
In Lemma~\ref{lemma:RBDS-MaxDegCont-forward} and Lemma~\ref{lemma:RBDS-MaxDegCont-backward} we prove the forward and backward directions, respectively.

\begin{lemma}\label{lemma:RBDS-MaxDegCont-forward} Suppose the reduction returns $(G, k, d)$ when the input is $(H, R, B, l)$. 
If $(H, R, B, l)$ is a \yes\ instance of \textsc{RBDS} then $(G, k, d)$ is a \yes\ instance of \textsc{MDC}.
\end{lemma}
\begin{proof}
As mentioned before, $(H, R, B, l)$ is a \yes\ instance of \textsc{RBDS} if and only if $(H^{\star}, R, B^1 \cup B^2, l)$ is a \yes\ instance of \textsc{RBDS}.
Let $R’ = \{r_1, r_2, \dots, r_{|R’|}\}$ be a subset of $R$ of size at most $l$ such that $B^1 \cup B^2 = N_{H^{\star}}(R)$. 
Without loss of generality, we assume that $R’$ is a minimal dominating set.
Partition $B^1 \cup B^2$ into $B_1, B_2, \dots, B_{|R’|}$ such that for every $j \in [|R’|]$, $r_j$ dominates vertices in $B_j$ and for every $b$ in $B$, vertices $b^1, b^2$ are in same part.
Since $R’$ is a minimal dominating set, $B_j$ is a non-empty.
For every $j \in [|R’|]$, define $F_j$ as follows: Initialize $F_j$ to an empty set.
For every $b$ in $B_j$, the consider binary tree rooted at $b^1$ (similarly at $b^2$).
Add the root-to-leaf paths in the trees that correspond to edges $(r_j, b^1)$ and $(r_j, b^2)$ to $F_j$.
Let $F$ be the union of all $F_j$s.
Formally, $F = \bigcup_{j \in [|R’|]} F_j$.
Since every path is of length $\log_2|R|$ and any two paths in $F$ are edge-disjoints, $|F| = 2|B| \cdot \log_2|R|$.
Consider the graph $G/F$ and let $G$ is contracted to $G/F$ via function $\psi$. 
Define $B^{\circ} = \{b^{\circ}_1, b^{\circ}_2, \dots, b^{\circ}_{|R’|}\}$, where $b^{\circ}_j$ is the vertex in $G/F$ which is obtained by contracting all edges in $F_j$.
We argue that the maximum degree of vertices in $G/F$ is at most $d$.

Vertices in $V(G/F)$ can be partitioned into $R \setminus R’ = R \setminus V(F)$, $I \setminus V(F)$, $B^{\circ}, U, X$, and $P$.
Here, $I, U, X, P$ are the sets defined in Lemma~\ref{lemma:RBDS-Red-Prop}.
For every vertex in $(R \cup I \cup R) \setminus V(F)$, we have $|W(\psi(v))| = 1$.
By Observation~\ref{obs:witness-structure-property} (\ref{item:degree-bound}), $\deg_{G/F}(\psi(v)) \le \deg_G(v)$.
By Lemma~\ref{lemma:RBDS-Red-Prop}, every vertex in $R \cup I \cup R$ has degree at most $d$.
Hence, we can conclude that $\deg_{G/F}(\psi(v)) \le d$. 
In graph $G/F$, every vertex $v$ in $U$ is adjacent to every vertex in $B^{\circ}$ and $d - l$ pendant vertices. 
As $|B^{\circ}| = |R’| \le l$, every vertex in $U$ is adjacent to at most $d$ vertices.
As $b^1,b^2$ are in same witness set for every $b$ in $B$, every vertex in $X$ is adjacent to one vertex in $B^{\circ}$ and $d - 1$ pendent vertices.
Hence, every vertex in $X$ has degree at most $d$ in $G/F$.

It remains to argue that the degree of $b^{\circ}_j$ is at most $d$ in $G/F$.
For every $j$ in $[|R’|]$, set $V(F_j)$ can be partitioned into the following three parts: $(i) V(F_j) \cap R$, $(ii) V(F_j) \cap I$, and $(iii) B_j = V(F_j) \cap (B^1 \cup B^2)$.
Consider the first part. 
By construction, $V(F_j) \cap R = \{r_j\}$.
As mentioned in the proof of Lemma~\ref{lemma:RBDS-Red-Prop}, in graph $G$, any vertex in $R$ is adjacent to at most $|B^1 \cup B^2|$ vertices in $I$.
Hence, the vertex in $V(F_j) \cap R$ is adjacent to at most $|B^1 \cup B^2| - |B_j|$ vertices outside $V(F_j)$.
Now consider the second part.
Every vertex in $V(F_j) \cap I$ is adjacent to at most one vertex outside $V(F_j)$.
As every $r_j$ to $b^1$ (similarly $r_j$ to $b^2$) path is of length $\log_2|R|$, and any two paths in $F_j$ are edge-disjoints, $|V(F_j) \cap I| = |B_j| \cdot \log_2|R|$.
Hence, vertices in $V(F_j) \cap I$ are adjacent to at most $|B_j| \cdot \log_2|R|$ vertices outside $V(F_j)$.
Now consider the third part.
Every vertex in $V(F_j) \cap (B^1 \cup B^2)$ is adjacent to at most one vertex in $I \setminus V(F_j)$, every vertex in $U$ and every vertex in $X_{j}$.
Here, $X_{j} = \bigcup_{b^1 \in B_j} X_b = \bigcup_{b^2 \in B_j} X_b$.
Hence, vertices in $V(F_j) \cap (B^1 \cup B^2)$ are adjacent to at most $|B_j| + (k + 1) + |B_j|(k + 1)$.
This implies that the number of vertices adjacent to $V(F_j)$ is at most $|B^1 \cup B^2| - |B_j| + |B_j| \cdot \log_2|R| + |B_j| + (k + 1) + |B_j| \cdot (k + 1) \le 2 |B^1 \cup B^2| \cdot (\log_2|R| + k + 2) = d$.
Here, we use the fact that $1 + |B_j| \le |B^1 \cup B^2|$.
This follows from our assumption that in graph $H$, no vertex in $R$ is adjacent to all but one vertex in $B$.
Hence, the degree of any vertex in $B^{\circ}$ is at most $d$ in $G/F$.

This prove that the maximum degree of any vertex in $G/F$ is at most $d$. 
Hence, if $(H, R, B, l)$ is a \yes\ instance, then so is $(G, k, d)$.
\end{proof}

We now prove the backward direction.
As in Section~\ref{sec:lower-bound}, we prove a series of claims about a solution to reduced instance.
We prove the five properties mentioned at the start of this section to prove the following lemma.

\begin{lemma}\label{lemma:RBDS-MaxDegCont-backward} Suppose the reduction returns $(G, k, d)$ when the input is $(H, R, B, l)$. 
If $(G, k, d)$ is a \yes\ instance of \textsc{MDC} then $(H, R, B, l)$ is a \yes\ instance of \textsc{RBDS}.
\end{lemma}

We prove that if $(G, k, d)$ is a \yes\ instance of \textsc{MDC} then $(H^{\star}, R, B^1 \cup B^2, l)$ is a \yes\ instance of \textsc{RBDS}.
Recall that for vertex subset $X, Y$, we denote the set of all edges with one endpoint in $X$ and another endpoint in $Y$ by $E(X, Y)$.
Let $P^1_U, P^2_U$ and $P_X$ be the collection of pendent vertices adjacent to vertices in $U^1, U^2$, and $X$, respectively. 
By construction, we can partition edges of $G$ into the following four sets: 
$E(B^1 \cup B^2, U^1 \cup U^2)$, $E(B^1 \cup B^2, X)$, $E(U \cup X, P^1_U \cup P^2_U \cup P_X)$, and $E’$.
Here, $E’$ is the collection of edges that are not covered by the first three sets.

Suppose $(G, k, d)$ is a \yes\ instance and $F$ is a solution to $(G, k, d)$.
\begin{claim} \label{claim:no-BU} $F \cap E(B^1 \cup B^2, U^1 \cup U^2) = \emptyset$.
\end{claim}
\begin{proof}
Assume that there is an edge, say $(b^1, u)$, in $F \cap E(B^1, U^1)$ where vertices $b^1, u$ are in $B^1$ and $U^1$, respectively.
Let $w$ be the new vertex introduced while contracting edge $(b^1, u)$.
In graph $(G/\{(b^1, u)\}$, vertex $w$ is adjacent to every vertex in $B^1 \setminus \{b^1\} \cup X_b \cup (U^1 \setminus \{u\})$ and with all pendent vertices that were adjacent with $u$ in $G$.
Hence, the degree of $w$ in $(G/\{(b^1, u)\})$ is at least $|B| - 1 + |X_b| + |U^1| + d - l > d + (k - 1) + 1$.
By Observation~\ref{obs:high-deg-vertex}, $(G/\{(b^1, u)\}, k - 1, d)$ is a \no\ instance. 
This contradicts Observation~\ref{obs:reduced-inst-yes}.
Hence our assumption is wrong and $F \cap E(B^1, U^1)$ is an empty set.
By similar arguments, $F \cap E(B^2, U^2)$ is an empty set.
By construction, sets $E(B^1, U^2)$ and $E(B^2, U^1)$ are empty. 
This concludes the proof of the claim.
\end{proof}
\begin{claim} \label{claim:no-BX} $F \cap E(B^1 \cup B^2, X) = \emptyset$. 
\end{claim}
\begin{proof}
To prove the claim, it suffices to prove that for any $b$ in $B$, $F$ does not include edge $(b^1, x)$ where $b^1$ is in $B^1$ and $x$ is in $X_b$.
For the sake of contradiction, assume such an edge exists.
Let $w$ be the new vertex introduced while contracting edge $(b^1, x)$.
In graph $(G/\{(b^1, x)\}$, vertex $w$ is adjacent to every vertex in $\{b^2\} \cup (X_b \setminus \{x\}) \cup U^1$ and with all pendent vertices that were adjacent with $x$ in $G$.
Hence, the degree of $w$ in $(G/\{(b^1, u)\})$ is at least $1 + |X_b| + |U^1| + d - 1 > d + (k - 1) + 1$.
By Observation~\ref{obs:high-deg-vertex}, $(G/\{(b^1, x)\}, k - 1, d)$ is a \no\ instance. 
This contradicts Observation~\ref{obs:reduced-inst-yes}.
Hence our assumption is wrong and there is no edge of the form $(b^1, x)$.
As every edge in $E(B^1, X)$ is of the form $(b^1, x)$ for some $b^1$ in $B^1$ and $x$ in $X_b$, we can conclude that $E(B^1, X)$ is an empty set.
By similar arguments, $E(B^3, X)$ is an empty set.
This concludes the proof of the claim.
\end{proof}
\begin{claim} \label{claim:B-in-F} $(B^1 \cup B^2) \subseteq V(F)$.
\end{claim}
\begin{proof} 
Assume for the sake of contradiction that there is $b^1$ in $(B^1 \cup B^2) \setminus V(F)$.
Recall that every vertex $x$ in $X_b$, the degree of $x$ is $d + 1$.
By Observation~\ref{obs:high-deg-vertex}, there are at least two vertex in $N[x]$ which are in $V(F)$.
As $b^1$ is not incident to any solution edge, $b^1$ is not in $N[x] \cap V(F)$.
Vertex $b^2$ can be one of the vertices in $N[x] \cap V(F)$.
By Claim~\ref{claim:no-BX}, edge $(b^2, x)$ is not in any solution.
Hence, there is at least one vertex $N[x] \cap V(F)$ which is a pendent vertex adjacent to $x$ or the vertex $x$ itself.
This implies for every $x$ in $X_b$, there is a solution edge incident to pendent vertex adjacent to $x$.
Hence, there are at least $|X_b| = k + 1$ edges in $F$.
This contradicts the fact that $|F|$ is at most $k$.
Hence our assumption is wrong and $B^1 \cup B^2 \subseteq V(F)$.
\end{proof}
\begin{claim}\label{claim:l-parts} There are at most $l$ witness sets in the $G/F$-witness structure of $G$ that contains vertices in $B^1$ (similarly in $B^2$).
\end{claim}
\begin{proof}
Recall that for a subset $Z \subseteq V(G)$, we define $\psi(Z) := \{\psi(z)\ |\ z \in X \}$.
To prove the claim, it suffices to prove that 
the size of $\psi(B^1)$ is at most $l$.
Assume there is an integer $l’ \ge 1$ such that $|\psi(B^1)| = l + l’$.
For every vertex $u \in U$ in graph $G$, vertex $\psi(u)$ is adjacent to every vertex in $\psi(B^1)$ in graph $G’$.
The degree of $\psi(u)$ in $G/F$ is at most $d$.
By Claim~\ref{claim:no-BU}, $F$ does not contain an edge in $E(B^1, U^1)$.
Hence, $F$ must contain at least $l’$ many edges incident to $u$ and pendent vertices adjacent to it.
As this is true for every vertex in $U^1$, there are $|U^1| \cdot l’$ many edges in $F$ that are incident to pendant vertices.
As $|U^1| = k + 1$ and $l’ \ge 1$, this contradicts the fact that $|F| \le k$.
Hence, our assumption is wrong and $|\psi(B^1)| \le l$.
This implies the $G/F$-witness structure of $G$ partitions all vertices in $B^1$ into at most $l$ witness sets.
By similar arguments, we can prove that the $G/F$-witness structure of $G$ partitions all vertices in $B^2$ into at most $l$ witness sets.
\end{proof}
\begin{claim} \label{claim:twins-same} For every $b$ in $B$, $\psi(b^1) = \psi(b^2)$.
\end{claim}
\begin{proof}
Assume for the sake of contradiction that there is $b$ in $B$, such that $\psi(b^1) \neq \psi(b^2)$.
In other words, $b^1$, $b^2$ are in two different witness sets in $G/F$-witness structure of $G$.
Recall that every vertex $x$ in $X_b$, the degree of $x$ is $d + 1$.
By Observation~\ref{obs:high-deg-vertex}, there are at least two vertex in $N[x]$ which are in sane witness set. 
By Claim~\ref{claim:no-BX}, no edge in $E(B^1 \cup B^2, X)$ is a part of any solution.
Hence, for every $x$ in $X_b$, there is a solution edge incident to pendent vertex adjacent to $x$.
As $|X_b| = k + 1$, this contradicts the fact that $|F|$ is at most $k$.
Hence our assumption is wrong and for every $b$ in $B$, $\psi(b^1) = \psi(b^2)$.
\end{proof}

We are now able to present a proof of Lemma~\ref{lemma:RBDS-MaxDegCont-backward}.
In the proof, we crucially use the fact that $G[R \cup I \cup B^1 \cup B^2]$ is a union of binary trees rooted at vertices in $B^1 \cup B^2$.
Moreover, any two of these binary trees are edge disjoint.

\begin{proof}(of Lemma~\ref{lemma:RBDS-MaxDegCont-backward})
We prove that $(H^{\star}, R, B^1 \cup B^2, l)$ is a \yes\ instance of \textsc{RBDS}.
By Claim~\ref{claim:B-in-F} and \ref{claim:l-parts}, there is $l’ (\le l)$ witness sets, say $W_1, W_2, \dots, W_{l’}$, in the $G/F$-witness structure of $G$ such that their union contains $B^1 \cup B^2$. 
For $j \in [l’]$, define $B^{\circ}_j = (B^1 \cup B^2) \cap W_{j}$.
We divide proof of the lemma in two parts.
First, we prove that for every $b_{\alpha}$ in $B^{\circ}_j$, there is vertex $r$ in $W_j$ such that $r$ is adjacent with $b_{\alpha}$ in $H^{\star}$.
This implies that in graph $H^{\star}$, set $\bigcup_{j \in [l’]}(R \cap W_j)$ dominates $B^1 \cup B^2$.
In the second part, we prove that there is at most one vertex in $R \cap W_j$.
This proves that the dominating set is of size $l’ \le l$.

Define $\widetilde{E} := E(B^1 \cup B^2, U^1 \cup U^2) \cup E(B^1 \cup B^2, X)$.
By Claim~\ref{claim:no-BU} and \ref{claim:no-BX}, solution $F$ does not contain any edge in $\widetilde{E}$.
For any two vertices $b_{\alpha}, b_{\beta} \in B^1 \cup B^2$, any $b_{\alpha}$ to $b_{\beta}$ path in $G - \widetilde{E}$ contains a vertex in $N_{H^{\star}}(b_{\alpha})$ and a vertex in $N_{H^{\star}}(b_{\alpha})$.
Fix an arbitrary vertex $b^1$ in $B^1$.
By Claim~\ref{claim:twins-same}, if $b^{1}$ is in $W_j$ then $b^{2}$ is also in $W_j$.
Hence, there are at least two vertices in $B^{\circ}_j$ which is a subset of $W_j$. 
As $W_j$ is connected set in $G - \widetilde{E}$, it contains at least one vertex each from $N_{H^{\star}}(b_{\alpha})$ and $N_{H^{\star}}(b_{\beta})$.
Hence, for any $b_{\alpha}$ in $B^1 \cup B^2$, if $b_{\alpha}$ is in $W_j$ then there exists at least one vertex in $W_j$ which is adjacent to $b_{\alpha}$ in $H^{\star}$.
This implies that in graph $H^{\star}$, set $\bigcup_{j \in [l’]}(R \cap W_j)$ dominates $B^1 \cup B^2$.

To prove the second part, we need to argue that we can partition $F$ into $|B^1 \cup B^2|$ parts, each corresponding to a vertex in $B^1 \cup B^2$.
For every vertex $b_{\alpha}$ in $B^1 \cup B^2$, let $\rho(b_{\alpha})$ be the vertex in $R$ such that $(i)$ $b_{\alpha}$ and $\rho(b_{\alpha})$ are in same witness set, and $(ii)$ $\rho(b_{\alpha})$ is the nearest vertex in the witness set that satisfy the first property. 
The arguments in previous the paragraph ensure that such vertex exits. 
Let $P_{\alpha}$ be a path from $b_{\alpha}$ to $\rho(b_{\alpha})$ such that edges in $P_{\alpha}$ are in $F$.
As $b_{\alpha}, \rho(b_{\alpha})$ are in $W_j$, which is a witness set, such a path exists.
Because of the second property, $P_{\alpha}$ is a root-to-leaf path in the binary tree rooted at $b_{\alpha}$.
This implies that the length $P_{\alpha}$ is $\log_2|R|$ and for two different vertices $b_{\alpha}, b_{\beta}$ in $B^1 \cup B^2$, paths $P_{\alpha}$ and $P_{\beta}$ are edge-disjoint.
As $|F| = |B^1 \cup B^2| \cdot \log_2|R|$, we can conclude that $\{E(P_{\alpha})\ |\ b_{\alpha} \in B_1 \cup B_2\}$ is a partition of $F$.

We now prove the second part. 
Assume that there is $W_j$ such that $j$, set $R \cap W_j$ contains two vertices, say $r_1, r_2$.
As $r_1, r_2$ are in same witness sets, there is a $r_1$ to $r_2$ path whose edges are contained in $F$.
By the construction of $H^{\star}$ and the fact that $F \cap \widetilde{E} = \emptyset$, there is vertex $b_{\alpha}$ in $W_j$ such that $r_1, r_2$ are leaves in the binary tree rooted at $b_{\alpha}$.
Without loss of generality, let $r_1 = \rho(b_{\alpha})$.
As $r_2$ is a leaf, there is at least one edge in $r_1$ to $r_2$ path which is not contained in path $P_{\alpha}$. 
This edge is not a part of $P_{\beta}$ for any $b_{\beta} \neq b_{\alpha} \in B^{1} \cup B^2$ as binary trees rooted at vertices in $B^1 \cup B^2$ are edge disjoints.
This implies there is an edge in $F$ which is not in path $P_{\alpha}$ for any $b_{\alpha}$ in $B^1 \cup B^2$.
This contradictions the fact that $\{E(P_{\alpha})\ |\ b_{\alpha} \in B_1 \cup B_2\}$ is a partition of $F$.
Hence our assumption is wrong and for $j \in [l’]$, set $R \cap W_j$ contains at most one vertex. 

This implies that in graph $H^{\star}$, set $\bigcup_{j \in [l’]}(R \cap W_j)$ is of size at most $l$ and dominates $B^1 \cup B^2$.
Hence $(H^{\star}, R, B^1 \cup B^2, l)$ is a \yes\ instance. 
As mentioned before, it is easy to see that $(H, R, B, l)$ is a \yes\ if and only if $(H^{\star}, R, B^1 \cup B^2, l)$ is a \yes\ instance.
This concludes the proof of the lemma. 
\end{proof}

We are now able to present a proof of Theorem~\ref{thm:no-poly-kernel-MDC} using Corollary~\ref{corr:no-poly-kernel-RBDS}, Lemma~\ref{lemma:RBDS-MaxDegCont-forward}, and Lemma~\ref{lemma:RBDS-MaxDegCont-backward}.

\begin{proof}(of Theorem~\ref{thm:no-poly-kernel-MDC})
Assume, for the sake of contradiction, \textsc{MDC}, parameterized by $k + d$ admits a polynomial-sized compression.
This implies there is an algorithm, say $\calA$, that given any instance $(G, k, d)$ of \textsc{MDC} runs in time polynomial time and produces equivalent instance $(I’, k’)$ of parameterized problem $\Pi$ such that 
$(i)$ $(G, k, d)$ is a \yes\ instance of \textsc{MDC} if and only if $(I’, k’)$ is a \yes\ instance of $\Pi$, and 
$(ii)$ $|I’| + k’ \le (k + d)^c$, where $c$ is a fixed constant.
We construct a compression algorithm for \textsc{RBDS} using Algorithm $\calA$ as a subroutine.

Consider Algorithm $\calB$ that given an instance $(H, R, B, l)$ of \textsc{RBDS} constructs an instance $(G, k, d)$ of \textsc{MDC} as described in the reduction.
Then, it calls Algorithm $\calA$, as subroutine, on instance $(G, k, d)$.
Let $(I’, k’)$ be the instance of $\Pi$ returned by Algorithm $\calA$.
The algorithm returns $(I’, k’)$ as a compressed instance.

The correctness of the algorithm $\calA$, Lemma~\ref{lemma:RBDS-MaxDegCont-forward}, and Lemma~\ref{lemma:RBDS-MaxDegCont-forward} implies that $(H, R, B, l)$ is a \yes\ instance of \textsc{RBDS} if and only if $(I’, k’)$ is a \yes\ instance of $\Pi$.
Since $(G, k, d)$ is the instance of \textsc{MDC} constructed by the reduction when input was $(H, R, B, l)$, we have $k = 2|B| \cdot \log_2|R|$ and $d = 2 |B| \cdot (\log_2|R| + k + 2)$.
As, $|I’| + k’ \le (k + d)^c$, we have $|I’| + k’ \le (|B| + \log_2|R|)^{c_0}$, where $c_0$ is a fixed constant.
By the description of the reduction, it is easy to see that given instance $(H, R, B, l)$, Algorithm $\calB$ computes instance $(G, k, d)$ in time polynomial in $|V(H)|$.
Hence, the total running time of the algorithm is polynomial in the size of the input.

This implies \textsc{RBDS}, when parameterized by $|B| + \log_2|R|$, admits a polynomial compression. But this contradicts Corollary~\ref{corr:no-poly-kernel-RBDS}.
Hence, our assumption was wrong, which concludes the proof.
\end{proof}
\section{Conclusion}
\label{sec:conclusion}
In this article, we studied \textsc{Maximum Degree Contraction} problem.
We prove that a simple brute force algorithm for this problem is optimal under \ETH.
This lower bound also implies that the known \FPT\ algorithm with running time $(d + k)^{k} \cdot n^{\calO(1)}$ is also optimal under the same hypothesis.
We compliment this result by presenting another \FPT\ algorithm with running time $2^{\calO(dk)} \cdot n^{\calO(1)}$.
While these two \FPT\ algorithms are incomparable, our algorithm runs faster for smaller values of $d$, for which the problem still remains \NP-\Hard.
We also prove that unless $\NP \subseteq \CONP/poly$, the problem does not admit a polynomial compression when parameterized by $k + d$.

Most of the $\calH$-\textsc{Contraction} problems do not admit a polynomial kernel under the same complexity conjecture.
For some graph classes like trees, cactus, cliques, splits graphs, such negative results have been complimented by establishing a \emph{lossy kernel} of polynomial size for these problems.
There are also examples like \textsc{Chordal Contraction}, \textsc{s-Club Contraction} (for $s \ge 2$) for which we know that lossy kernel of polynomial size do not exist.
We conclude this article with following open question:
Does \textsc{Maximum Degree Contraction} admit a lossy kernel of polynomial size?

%%
%% Bibliography
%%

%% Please use bibtex, 

\bibliography{references}

\end{document}